\documentclass{article}

\usepackage{microtype}
\usepackage{graphicx}
\usepackage{subfigure}
\usepackage{booktabs} 

\usepackage{hyperref}

\usepackage[accepted]{icml2025}

\usepackage{amsmath}
\usepackage{amssymb}
\usepackage{mathtools}
\usepackage{amsthm}

\usepackage[capitalize,noabbrev]{cleveref}

\theoremstyle{plain}
\newtheorem{theorem}{Theorem}[section]
\newtheorem{proposition}[theorem]{Proposition}
\newtheorem{lemma}[theorem]{Lemma}
\newtheorem{corollary}[theorem]{Corollary}
\theoremstyle{definition}
\newtheorem{definition}[theorem]{Definition}

\theoremstyle{remark}
\newtheorem{remark}[theorem]{Remark}
\usepackage[flushleft]{threeparttable}

\usepackage[textsize=tiny]{todonotes}

\newcommand{\cA}{\mathcal{A}}

\newcommand{\cG}{\mathcal{G}}

\newcommand{\cO}{\mathcal{O}}

\newcommand{\cR}{\mathcal{R}}

\newcommand{\cZ}{\mathcal{Z}}

\renewcommand{\epsilon}{\varepsilon}

\def\E{\mathop{\mathbb{E}}\displaylimits} 

\newcommand{\paren}[1]{\left( #1 \right)}
\newcommand{\parens}[1]{( #1 )}
\newcommand{\sqb}[1]{\left[ #1 \right]}

\newcommand{\set}[1]{\left\{#1 \right\}}
\newcommand{\sets}[1]{\{#1 \}}

\newcommand{\kets}[1]{\vert #1 \rangle}

\newcommand{\abss}[1]{\left\lvert #1 \right\rvert}
\newcommand{\abs}[1]{\lvert #1 \rvert}

\newcommand{\ket}[1]{\left| {#1}\right \rangle}

\newcommand{\norm}[1]{\left\lVert #1 \right\rVert}
\newcommand{\norms}[1]{\lVert #1 \rVert}

\newcommand{\ang}[1]{\left\langle #1 \right\rangle}
\newcommand{\angs}[1]{\langle #1 \rangle}

\newcommand{\ceils}[1]{\lceil #1 \rceil}

\newcommand{\diag}{{\operatorname{diag}}}
\newcommand{\trace}{{\operatorname{Tr}}}
\newcommand{\poly}{{\operatorname{poly}}}
\newcommand{\polylog}{{\operatorname{polylog}}}

\icmltitlerunning{Quantum Speedup for Hypergraph Sparsification}

\begin{document}

\twocolumn[
\icmltitle{Quantum Speedup for Hypergraph Sparsification}




\begin{icmlauthorlist}
\icmlauthor{Chenghua Liu}{ios,ucas}
\icmlauthor{Minbo Gao}{ios,ucas}
\icmlauthor{Zhengfeng Ji}{thu}
\icmlauthor{Mingsheng Ying}{uts}
\end{icmlauthorlist}

\icmlaffiliation{ios}{Institute of Software, Chinese Academy of Sciences, China}
\icmlaffiliation{ucas}{University of Chinese Academy of Sciences, China}
\icmlaffiliation{thu}{Department of Computer Science and Technology, Tsinghua University, China}
\icmlaffiliation{uts}{University of Technology Sydney, Australia}
\icmlcorrespondingauthor{Zhengfeng Ji}{jizhengfeng@tsinghua.edu.cn}

\icmlkeywords{Machine Learning, ICML}

\vskip 0.3in
]



\printAffiliationsAndNotice{}  

\begin{abstract}
  Graph sparsification serves as a foundation for many algorithms, such as
  approximation algorithms for graph cuts and Laplacian system solvers.
  As its natural generalization, hypergraph sparsification has recently gained
  increasing attention, with broad applications in graph machine learning and
  other areas.
  In this work, we propose the \emph{first} quantum algorithm for hypergraph
  sparsification, addressing an open problem proposed
  by~\citet{apers2022quantum}.
  For a weighted hypergraph with $n$ vertices, $m$ hyperedges, and rank $r$, our
  algorithm outputs a near-linear size $\epsilon$-spectral sparsifier in time
  $\widetilde O\parens{r\sqrt{mn}/\epsilon}$\footnote{We use
  $\widetilde O\parens{f}$ to represent
  $O \paren{f \cdot \polylog\parens{m, n, r, 1/\epsilon}}$ throughout this
  paper to suppress polylogarithmic factors.}.
  This algorithm matches the quantum lower bound for constant $r$ and
  demonstrates quantum speedup when compared with the state-of-the-art
  $\widetilde O(mr)$-time classical algorithm.
  As applications, our algorithm implies quantum speedups for computing
  hypergraph cut sparsifiers, approximating hypergraph mincuts and hypergraph
  $s$-$t$ mincuts.
\end{abstract}

\section{Introduction}

Sparsification serves as a foundational algorithmic paradigm wherein a densely
constructed entity is transitioned to a sparse counterpart while preserving its
inherent characteristics.
Such a process invariably enhances various facets of algorithmic efficiency,
from reduced execution time to optimized space complexity, and even more
streamlined communication.
A typical instance of this paradigm is the graph sparsification, where the
objective is to reduce the number of edges by adjusting edge weights while
(approximately) preserving the spectral properties of the original graph.
Over the past two decades, graph sparsification has experienced several
breakthroughs (see~\citet{spielman2011graph,batson2014twice}), ultimately
leading to the discovery of algorithms that finds spectral sparsifiers of linear
size in nearly-linear time~\citep{lee2018constructing}.
Graph sparsification algorithms are used for crucial tasks such as Laplacian
system solver~\cite{cohen2014solving}, computing random walk
properties~\cite{cohen2016faster}, and solving maximum flow
problem~\cite{chen2022maximum}.
Additionally, they have found widespread applications in machine learning
domains, including computer vision~\cite{simonovsky2017dynamic},
clustering~\cite{peng2015partitioning,laenen2020higher,agarwal2022sublinear},
and streaming machine learning algorithms~\cite{braverman2021adversarial}.

Hypergraphs, a generalization of graphs, naturally emerge in various real-world
scenarios where interactions go beyond pairwise relationships, such as group
dynamics in biochemistry, social networks, and trade networks, as they allow a
single edge to connect any number of vertices (the max number of vertices an
edge contains is called  the rank).
Thus, as a natural generalization of the garph sparsification, the task of
sparsifying hypergraphs are investigated by many researchers, starting
from~\citet{soma2019spectral}.
Hypergraph sparsification significantly reduces the computational cost of
calculating hypergraph energy, a crucial quantity for many machine learning
tasks, including clustering~\citep{zhou2006learning, hein2013total,
  takai2020hypergraph}, semi-supervised learning~\citep{hein2013total,
  zhang2018re, yadati2019hypergcn, li2020quadratic}, and link
prediction~\citep{yadati2020nhp}.

With the rapid development of quantum computing, many graph algorithms and
machine learning tasks have benefited from quantum speedups.
In the context of the sparsification paradigm, the groundbreaking work of
\citet{apers2022quantum} introduced the \emph{first quantum} speedup for graph
sparsification with nearly optimal query complexity, showcasing the potential of
quantum computing for sparsification tasks.
They further proposed several open questions about the potential for quantum
speedups in broader sparsification tasks.
Among these, a natural and important question is:
\begin{quote}
  \textit{Is there a hypergraph sparsification algorithm that enables quantum
  speedups?}
\end{quote}

In this paper, we give an affirmative answer to this question. 
Specifically, we develop a quantum algorithm for constructing an
$\varepsilon$-spectral sparsifier of near-linear size for a weighted hypergraph
with $n$ vertices, $m$ hyperedges, and rank $r$.
Our algorithm achieves near-linear output size—matching the current best
classical algorithm—while operating in sublinear time
$\widetilde O\parens{r\sqrt{mn}/\epsilon}$\footnote{To be more precise, our
  algorithm runs in $\widetilde O\parens{r\sqrt{mn}/\epsilon+r\sqrt{mnr}}$ time,
  of which the second term is usually smaller in most situations.
  See~\cref{re:main-theorem-1} for more detailed discussion.}, which improves
the time complexity $\widetilde{O}(mr)$ of the best known classical
algorithm\footnote{Typically, we assume that $\epsilon \geq \sqrt{n/m}$, as
  sparsification is only advantageous when the number of hyperedges in the
  sparsifier is at most $m$.}.
For the constant rank $r$, this time complexity matches the quantum lower bound
of $\widetilde{\Omega}(\sqrt{mn}/\varepsilon)$ established
by~\citet{apers2022quantum} and contrasts with the classical lower bound of
$\Omega(m)$ (see~\cref{re:main-theorem-2}).
Additionally, for dense hypergraphs, where $m \in \Omega\parens{n^r}$, our
algorithm achieves near-quadratic improvement, reducing the time complexity from
$\widetilde O\parens{rn^r}$ classically to $\widetilde O\parens{rn^{(r+1)/2}}$
quantumly.

\paragraph{Hypergraph Sparsification}
To extend the graph sparsification task to hypergraphs, we need to generalize
the quadratic form of the graph Laplacian.
This generalization leads to the concept of hypergraph energy, first introduced
by \citet{hubert2018spectral,yoshida2019cheeger}.
For a hypergraph $H = \paren{V,E,w}$, the energy of a vector $x\in \mathbb{R}^V$
is defined as a sum over all hyperedges $e \in E$, where each term is the
product of the edge weight $w_e$ and the ``quadratic form'' $Q_e(x)$.
Here, $Q_e(x)$ represents the \emph{maximum} squared difference between any two
vector components $x_u$ and $x_v$ corresponding to vertices $u$ and $v$ in the
hyperedge $e$ (see \cref{def:energy} for a formal description).
The concept of energy captures the spectral properties of the hypergraph, while
its inherent nonlinear structure introduces significant computational
challenges~\cite{hubert2018spectral}.

The task of hypergraph sparsification aims to reduce the number of hyperedges
while maintaining the energy of the original hypergraph, ultimately producing a
hypergraph sparsifier (see~\cref{def:hypergraph-spectral-sparsifier}).
Hypergraph sparsification not only is of great theoretical interest, but also
has wide applications in many fields, especially in graph machine learning.
Consequently, researchers have been developing increasingly sophisticated and
analytically refined algorithms for hypergraph sparsification in recent years.
\citet{soma2019spectral} were the first to demonstrate that an
$\epsilon$-spectral sparsifier of size $\widetilde{O}(n^3 / \epsilon^2)$ could
be constructed in time $\widetilde{O}(nmr + n^3 / \epsilon^2)$.
Subsequently, \citet{bansal2019new} improved the sparsifier size to
$\widetilde{O}(r^3n / \epsilon^2)$ with a construction time of
$\widetilde{O}(mr^2 + r^3n / \epsilon^2)$.
Further advancements by \citet{kapralov2021towards} and
\citet{kapralov2022spectral} reduced the sparsifier size to nearly linear,
achieving $\widetilde{O}(n / \epsilon^4)$ in polynomial time.
The current state-of-the-art algorithm for hypergraph sparsification, proposed
by \citet{jambulapati2023chaining}, produces a sparsifier of size
$\widetilde O\parens{n/\epsilon^2}$ in almost linear time $\widetilde{O}(mr)$.
Independently and concurrently, \citet{lee2023spectral} presented a
polynomial-time algorithm achieving a sparsifier of the same size.
For a detailed comparison of these results, see
\cref{tab:summary-hypergraph-sparsification}.

\begin{table*}[ht]
  \centering
  \caption{Summary of results on hypergraph sparsification}%
  \label{tab:summary-hypergraph-sparsification}
  \begin{threeparttable}[b]
  \begin{tabular}{cccc}
    \toprule
    Reference & Type & Sparsifier size & Time Complexity \\\toprule
    \citet{soma2019spectral} & Classical & $O(n^3\log n/\epsilon^{2})$
        & $\widetilde O(mnr+n^3/\epsilon^2 )$ \\\toprule
      \citet{bansal2019new} & Classical & $O(r^3n\log n /\epsilon^{2})$
        & $\widetilde O\parens{mr^2+r^3n /\epsilon^2}$ \\\toprule
    \citet{kapralov2021towards} & Classical & $nr{(\log n /\epsilon)}^{O(1)}$
        & $ O(mr^2)+ n^{O(1)}$ \\ \toprule
    \citet{kapralov2022spectral} & Classical & $O(n\log^3 n /\epsilon^{4})$
        & $\widetilde O\parens{mr+\poly\parens{n}}$ \\ \toprule
    \citet{jambulapati2023chaining,lee2023spectral} & Classical
        & $O( n \log n \log r /\epsilon^2) $
        & $ \widetilde O(mr)$\tnote{$\dagger$}   \\ \toprule
    This work & Quantum & $O\parens{n \log n \log r /\epsilon^2} $
        & $ \widetilde O\parens{r\sqrt{mnr} +r\sqrt{mn}/\epsilon}$ \\ \toprule
    \toprule
  \end{tabular}
  \begin{tablenotes}
    \item [$\dagger$] This $\widetilde O(mr)$ complexity corresponds to the
    algorithm proposed in~\citet{jambulapati2023chaining}.
    \end{tablenotes}
  \end{threeparttable}
\end{table*}

Prior to the emergence of spectral sparsification, early research focused on a
relatively weaker notion called the cut sparsification.
In the domain of hypergraphs, extensive research has been conducted on
hypergraph cut sparsifiers, yielding significant theoretical and practical
advances~\citep{kogan2015sketching, chekuri2018minimum}.
These sparsifiers have proven particularly valuable in efficiently approximating
cut minimization problems in hypergraphs, facilitating applications across
multiple domains.
Notable applications include VLSI circuit partitioning~\citep{alpert1995recent,
  karypis1997multilevel}, optimization of sparse matrix
multiplication~\citep{akbudak2013hypergraph, ballard2016hypergraph}, data
clustering algorithms~\citep{li2017inhomogeneous, liu2021strongly}, and ranking
data analysis~\citep{li2017inhomogeneous}.

\paragraph{Main Results}
In this work, we propose the first quantum algorithm for hypergraph
sparsification that produces a sparsifier of size
$\widetilde{O}(n/\varepsilon^2)$ in $\widetilde{O}(r\sqrt{mn}/\varepsilon)$
time, which breaks the linear barrier of classical algorithms.

\begin{theorem}[Informal version of~\cref{thm:quantum-hypergraph-sparsification}]
  There exists a quantum algorithm that, given query access to a hypergraph
  $H =\paren{V,E,w}$ with $\abss{V}=n$, $\abss{E}=m$, $w \in \mathbb{R}^E_{\geq 0}$,
  rank $r$ and $\epsilon> 0$, outputs with high probability\footnote{Throughout
  this paper, we say something holds ``with high probability'' if it holds
  with probability at least $1-O\paren{1/n}$.}
  an $\epsilon$-spectral sparsifier of $H$ with
  $\widetilde{O} \parens{n/\epsilon^2}$ hyperedges, in time
  $\widetilde O\parens{r\sqrt{mn}/\epsilon}$.
\end{theorem}

As a corollary, our algorithm could be used to construct a cut sparsifier for a
hypergraph in sublinear time.
This enables quantum speedups for
approximating hypergraph mincuts and $s$-$t$ mincuts, achieving sublinear time complexity with respect to the number of hyperedges. For further details, we refer readers to \cref{sec:applications}.

\paragraph{Techiques}
Our algorithms are inspired by a sampling-based framework appearing in classical
hypergraph sparsification algorithms.
The overall idea behind the framework is to compute a proper importance weight
for each hyperedge, and sample the hyperedges based on the weights.

Specifically, we adopt the method proposed in~\citet{jambulapati2023chaining},
where the weight on each hyperedge is set to be the group leverage score
overestimate (which we call hyperedge leverage score overestimate in our paper
to avoid ambiguity).
Classically, computing these overestimates would require $\widetilde{O}(mr)$
time by an iterative/contractive algorithm which mainly follows the algorithm
for computing an approximate John ellipse~\cite{cohen2019near}.
Then, one can sample $\widetilde{O}(n/\varepsilon^2)$ hyperedges in
$\widetilde{O}(n/\varepsilon^2)$ time, and reweight them to get a hypergraph
sparsifier.

We discover that, with the assistance of a series of classical and quantum
techniques, the computation of hyperedge leverage score overestimate can be
accelerated.
To be more precise, we realize that the computation of hyperedge leverage score
overestimate could be executed in a sequence of sparse underlying graphs
(see~\cref{def:sparse-underlying-graph} for more details), and these sparse
underlying graphs can be efficiently constructed (in
$\widetilde{O}(r\sqrt{mnr})$ time) by the quantum graph sparification algorithm
proposed by~\citet{apers2022quantum}.
Thus, we obtain a quantum algorithm (\cref{alg:QHLSO}) that allows efficient
\emph{queries} to the hyperedge leverage score overestimate running in sublinear
time.

The quantum procedure described above, however, introduces a new challenge for
the sampling step: the exact sampling probabilities cannot be directly accessed.
To resolve this issue, we utilize the technique of ``preparing many copies of a
quantum state''~\cite{hamoudi2022preparing} to get
$\widetilde{O}(n/\varepsilon^2)$ samples in
$\widetilde{O}(r\sqrt{mn}/\varepsilon)$ time without explicitly computing the
normalization constant for the sampling probability.
After sampling hyperedges, we encounter a problem in the reweighting stage due
to the unknown normalization constant.
We address this issue partially using the quantum sum estimation procedure
(\cref{thm:quantum-sum-estimate}), while the imprecision introduced during this
procedure necessitates a more rigorous analysis.
We complete the correctness analysis by employing the novel chaining argument
proposed in~\citet{lee2023spectral}.

\paragraph{Related Works}
Quantum algorithms have demonstrated significant potential in graph-theoretic
and optimization tasks through the sparsification paradigm, offering both
theoretical advances and practical applications.

The seminal work of \citet{apers2022quantum} established quantum algorithms for
graph sparsification, demonstrating speedups in cut approximation, effective
resistance computation, spectral clustering, and Laplacian system solving, while
also providing fundamental lower bounds for quantum graph sparsification.
This work catalyzed several breakthrough results for graph problems.
\citet{apers2020quantum} developed a quantum algorithm for exact minimum cut
computation, achieving quantum speedups when the graph's weight ratio is
bounded.
\citet{apers2021sublinear} extended these techniques to solve the exact minimum
$s$-$t$ cut problem.
The versatility of quantum graph sparsification was further demonstrated in
\citet{cade2023quantum}, where it enabled accelerated motif clustering
computations.

A significant theoretical advancement emerged in \citet{apers2023quantum}, which
generalized the quantum graph sparsification framework to quantum spectral
approximation by combining leverage score sampling with Grover search.
This generalization enabled efficient approximation of Hessians and gradients in
barrier functions for interior point methods, yielding quantum speedups for
linear programming under specific conditions.
These spectral approximation techniques have found recent applications in
various machine learning problems.
\citet{song2023revisiting} and \citet{li2024quantum} applied the quantum
spectral approximation and leverage score sampling algorithms to achieve quantum
advantages in linear regression and John ellipsoid approximation, respectively.

\section{Preliminaries}

\subsection{Notation}

For clarity, we use $\sqb{n}$ to represent the set $\sets{1, 2, \ldots, n}$ and
$\sqb{n}_{0}$ to represent the set $\sets{0, 1, \ldots, n-1}$. 
Throughout this paper, we use $G = (V, F, c)$ to denote a graph, where $V$ is
the vertex set, $F$ is the edge set, and $c : F \to \mathbb{R}_{\geq 0}$ represents
the edge weights.
For an undirected weighted hypergraph, we denote it as $H = (V, E, w)$, where
$V$ is the vertex set, $E$ is the hyperedge set, and $w : E \to \mathbb{R}_{\geq0}$
represents the hyperedge weights.
We use $n,m$ to denote the size of $V$ and $E$ respectively.
Typically, we denote edges (consisting of two vertices) with $f$ and $g$, while
reserving $e$ for hyperedges.
Given a hyperedge $e$, we use $\binom{e}{2}$ to denote the corresponding induced
edge set $\sets{f \subseteq e : |f| = 2}$.

\subsection{Laplacian and Graph Sparsification}
For an undirected weighted graph $G = (V, F, c)$, the weighted degree of vertex
$i$ is defined by 
\begin{equation*}
\deg(i) := \sum_{f \in F: i \in f} c_f,
\end{equation*}
where the sum is
taken over all edges $f$ that contain $i$, and $c_f$ represents the weight of
edge $f$.

\begin{definition}[Laplacian]
  The Laplacian of a weighted graph $G=\parens{V,F,c}$ is defined as the matrix
  $L_{G}\in \mathbb{R}^{V\times V}$ such that
  \begin{equation*}
  \parens{L_ G}_{ij}=
  \begin{cases}
    \deg\parens {i} & \textup{if}i=j, \\
    -c_ {ij} & \textup{if}\sets{i,j}\in F, \\
    0 & \textup{otherwise}.
  \end{cases}
  \end{equation*}
\end{definition}

The Laplacian of graph $G$ is given by $L_G= D_G-A_G$, with $A_G$ the
weighted adjacency matrix $\paren{A_G}_{ij}=c_{ij}$ and $D$ the diagonal
weighted degree matrix $D_G=\diag\paren{\deg\paren{i}: i \in V}$.
$L_G$ is a positive semidefinite matrix whenever weight function $c$ is
nonnegative.
The quadratic form of $L_G$ can be written as
\begin{equation}
  \label{eq:quadratic-laplacian}
  x^\top L_ G  x =
  \sum_{\set {i,j}\in  F}c_{ij}\cdot \parens{x_i -x_j}^2
\end{equation}
for arbitrary vector $x \in \mathbb{R}^V$.
Graph sparsification produces a reweighted graph with fewer edges, known as a
graph (spectral) sparsifier.
A graph spectral sparsifier of $G$ is a re-weighted subgraph that closely
approximates the quadratic form of the Laplacian for any vector $x\in \mathbb{R}^V$.

\begin{definition}[Graph Spectral Sparsifier]
  Let $ G=\parens{V,F,c}$ be a weighted graph.
  A re-weighted graph $\widetilde G =\parens{V,\widetilde F, \widetilde c}$ is a
  subgraph of $G$, where $\widetilde c:\widetilde F\to \mathbb{R}_{\geq 0}$ and
  $\widetilde F= \sets{f\in F: \widetilde c_f>0}$.
  For any $\epsilon >0$, $\widetilde G $ is an $\epsilon$-spectral sparsifier of
  $G$ if for any vector $x\in \mathbb{R}^V$, the following holds:
  \begin{equation*}
  \abss{x^\top  L_ {\widetilde G}x  -x^\top L_ {G}  x} \leq
  \epsilon \cdot x^\top L_ {G}  x.
  \end{equation*}
\end{definition}

In the groundbreaking work by~\citet{spielman2011graph}, the authors
demonstrated that graphs can be efficiently sparsified by sampling edges with
weights roughly proportional to their effective resistances.
This importance sampling approach is foundational to graph sparsification and
has also inspired advancements in hypergraph sparsification.
Next, we define the effective resistance.

\begin{definition}[Effective Resistance]
  Given a graph $G =\parens {V,F, c}$, the effective resistance of a pair of
  $i , j \in V $ is defined as
  \begin{equation*}
  R_ {ij} := \parens {\delta_i -\delta_j}^\top L_G^{+} \parens {\delta_i - \delta_j}
  = \norms {L_G^{+/2} \parens {\delta_i -\delta_j}}^2,
  \end{equation*}
  where $L^{+}_G$ denotes the Moore-Penrose inverse of $L_G$, and $\delta_i$
  denotes the vector with all elements equal to 0 except for the $i$-th being 1.
\end{definition}
For further details, including key properties of effective resistance used in
this work, we refer the readers
to~\cref{sec:properties-of-effective-resistance}.

\subsection{Hypergraph Sparsification}
Here, we formally define the fundamental concept in hypergraph sparsification,
namely the energy.
\begin{definition}[Energy]\label{def:energy}
  Let $H= (V, E, w)$ be a weighted hypergraph.
  For every vector $x \in \mathbb{R}^ V$, we define its associated energy in $H$
  as
  \begin{equation}\label{eq:energy}
  Q_H (x) := \sum_{e \in E} w_e \cdot Q_e(x),
  \end{equation}
  where $Q_e \parens{x}:=\max_{\{i,j\} \subseteq e}{(x_i - x_j)}^2$.
\end{definition}

In the special case when the rank of $H$ is $2$, (meaning $H$ is actually a
graph), the energy reduces to the quadratic form of graph Laplacian (see
\cref{eq:quadratic-laplacian}).

Similar to graph sparsification, the goal of hypergraph sparsification is to
produce a hypergraph spectral sparsifier with fewer hyperedges.
The hypergraph spectral sparsifier of hypergraph $H$ is a reweighted subgraph of
$H$ that approximately preserves the energy for any vector $x \in \mathbb{R}^V$.

\begin{definition}[Hypergraph Spectral Sparsifier]\label{def:hypergraph-spectral-sparsifier}
  Let $ H=\parens{V,E,w}$ be a weighted hypergraph.
  A re-weighted hypergraph $\widetilde H= (V, \widetilde E, \widetilde w )$ is a
  subgraph of $H$, where $\widetilde w:\widetilde E\to \mathbb{R}_{\geq 0}$ and
  $\widetilde E= \sets{e\in E: \widetilde w_e>0}$.
  For any $\epsilon >0$, $\widetilde H $ is an $\epsilon$-spectral sparsifier of
  $G$ if for any vector $x\in \mathbb{R}^V$, the following holds:
  \begin{equation*}
    \left| Q_H (x )-Q_{\widetilde H}(x )\right| \leq \epsilon\cdot  Q_H (x ).
  \end{equation*}
\end{definition}

The hypergraph cut sparsifier is a weaker notion of sparsification than the
spectral sparsifier.
Specifically, for a weighted hypergraph $H \paren{V,E,w}$, we restrict
$x \in \mathbb{R}^V$ to be the characteristics vector $1_S \in \sets{0,1}^V$ of a
vertex subset $S \subseteq V$.
The energy $Q_H\paren{1_S}$, or simply $Q_H\parens{S}$, can be expressed as
$Q_H\parens{S}=\sum_{e\in \delta_S}w_e$, where $\delta_S$ denotes the set of
hyperedges crossing the cut $\parens{S,V\setminus S}$.
The \emph{${\epsilon}$-cut sparsifier} $\widetilde H$ of the hypergraph $H$ is
a subgraph that satisfies the following:
\begin{equation}\label{eq:hypergraph-cut}
  \abss{Q_H\paren{S} - Q_{\widetilde H}\paren{S}}
  \leq \epsilon \cdot Q_H \paren{S}, \quad \forall S \subseteq V.
\end{equation}

\subsection{Quantum Computing and Speedup}
In quantum mechanics, a $d$-dimensional quantum state
$\ket{v}=\paren{v_0,\ldots, v_{d-1}}^\top$ is a unit vector in a complex Hilbert
space $\mathbb{C}^d$, namely, $\sum_{i\in \sqb{d}_0} \abss{v_i}^2=1$.
We define the computational basis of the space $\mathbb{C}^d$ by
$\sets{\ket{i}}_{i \in \sqb{d}_0}$, where
$\ket{i}=\parens{0,\ldots,0,1,0,\ldots,0}^{\top}$ with the $i$-th entry
(0-indexed) being 1 and others being 0.
The inner product of quantum states $\ket{u},\ket{v}\in \mathbb{C}^d $ is
defined by $\angs{u | v}=\sum_{i \in \sqb{d}_0}u_i^* v_i$, where $z^*$ denotes
the conjugate of $z \in \mathbb{C}$.
The tensor product of quantum states $\ket{u}\in \mathbb{C}^{d_1}$ and
$\ket{v}\in \mathbb{C}^{d_2}$ is their Kronecker product,
$\ket{u}\otimes\ket{v}=\parens{u_0v_0,u_0v_1,\ldots ,u_{d_1-1}v_{d_2-1}}^{\top}$,
which can be abbreviated as $\ket{u}\ket{v}$.

A quantum bit, or qubit, is a quantum state $\ket{\psi}$ in $\mathbb{C}^2$,
expressible as $\ket{\psi}=\alpha\ket{0}+\beta\ket{1}$, where
$ \alpha,\beta\in \mathbb{C}$ and $\abs{\alpha}^2+\abs{\beta}^2=1$.
Furthermore, an $n$-qubit state is in the tensor product space of $n$ Hilbert
spaces $\mathbb{C}^2$, denoted as
$\parens {\mathbb{C}^2}^{\otimes n}=\mathbb{C}^{2^n}$, with the computational
basis $\sets{\ket{i}}_{i \in \sqb{2^n}_0}$.
To extract classical information from an $n$-qubit state $\ket{\psi}$, we
measure it in the computational basis, yielding outcome $i$ with probability
$p \paren{i}=\abs{\angs{\psi |i}}^2$ for $i \in \sqb{2^n}_0$.
The operations in quantum computing are described by unitary matrices $U$,
satisfying $U U^\dagger =U^{\dagger}U =I$, where $U^\dagger$ is the
Hermitian conjugate of $U$, and $I$ is the identity matrix.

We consider the following edge-vertex incidence oracle $\cO_G$ for the graph
$G= (V, F, c)$ with $n$ vertices and $m$ edges.
This oracle consists of two unitaries, $\cO_G^{\textup{vtx}}$ and
$\cO_G^{\textup{wt}}$, which are defined as follows for any edge
$f=\sets{i,j} \in F$:
\begin{align*}
  \cO_G^{\textup{vtx}} & :\ket{f}\ket{0}\mapsto\ket{f}\ket{i}\ket{j}, \\
  \cO_G^{\textup{wt}} & :\ket{f}\ket{0}\mapsto\ket{f}\ket{c_f},
\end{align*}
where $\ket{f}\in \mathbb{C}^m, \ket{i},\ket{j}\in \mathbb{C}^n$, and $c_f$ is
represented as a floating-point number with
$\ket{c_f} \in \mathbb{C}^{d_{\textup{acc}}}$.
Taking $d_{\textup{acc}}= \widetilde O\parens{1}$ allows for achieving arbitrary
desired floating-point accuracy.
Similarly, we assume access to hyperedge oracle $\cO_H $ for the hypergraph
$H=\paren{V,E,w}$, which consists of three unitaries $\cO_H^{\textup{size}}$,
$\cO_H^{\textup{vtx}}$ and $\cO_H^{\textup{wt}}$.
These unitaries allow for the following queries for any hyperedge $e \in E $:
\begin{align*}
  \cO_H^{\textup{size}} & :\ket{e}\ket{0}\mapsto\ket{e}\ket{\abss{e}},\\
  \cO_H^{\textup{vtx}} & :\ket{e}\ket{0}^{\otimes r} \mapsto\ket{e}
  \Bigl(\bigotimes_{i\in e} \ket{i} \Bigr)\ket{0}^{\otimes \parens {r-\abs{e}}},\\
  \cO_H^{\textup{wt}} & :\ket{e}\ket{0}\mapsto\ket{e}\ket{w_e}.
\end{align*}

In many quantum algorithms, information can be stored and retrieved in
quantum-read classical-write random access memory (QRAM)
(\citet{giovannetti2008quantum}), which is employed in numerous time-efficient
quantum algorithms.
QRAM enables the storage or modification of an array $c_1, \ldots, c_n$ of
classical data while allowing quantum query access via a unitary
$U_{\textup{QRAM}}: \ket{i}\ket{0} \mapsto \ket{i}\ket{c_i}$.
Although QRAM is a natural quantization of the classical RAM model and is widely
utilized, it is important to acknowledge that, given the current advancements of
quantum computers, the feasibility of implementing practical QRAM remains
somewhat speculative.

A quantum (query) algorithm $\cA$ is a quantum circuit consisting of a sequence
of unitaries $U_1, \ldots, U_T$, where each $U_t$ could be a quantum gate, a
quantum oracle, or a QRAM operation.
The time complexity of $\cA$ is determined by the number $T$ of quantum gates,
oracles and QRAM operations it contains.
The algorithm $\cA$ operates on $n$ qubits, starting with the initial state
$\ket{0}^{\otimes n}$.
The unitary operators $U_1, \dots, U_T$ are then applied sequentially to the
quantum state, resulting in the final quantum state
$\ket{\psi} = U_T \dots U_1 \ket{0}^{\otimes n}$.
Finally, a measurement is performed on $\ket{\psi}$ in the computational basis
$\ket{i}$ for $i \in \sqb{2^n}_0$, yielding a classical outcome $i$ with
probability $\abss{\angs{i | \psi}}^2$.

In our paper, we incorporate quantum algorithms from previous research as basic
components of our algorithm.
One such algorithm is quantum graph sparsification, initially proposed
by~\citet{apers2022quantum} using adjacency-list input queries, and later
revisited through alternative techniques in~\citet{apers2023quantum} using
edge-vertex incidence queries.
The query access described below refers to the latter approach.

\begin{theorem}[Quantum Graph Sparsification, Theorem 1 in~\citet{apers2022quantum}]\label{thm:apers-spectral-sparse}
  There exists a quantum algorithm $\mathsf{GraphSparsify}(\cO_G, \varepsilon)$
  that, given query access $\cO_G$ to a weighted graph $G=\paren{V,F,c}$ with
  $\abss{V}=n, \abss{F}=m, c \in \mathbb{R}^F_{\geq 0}$ and
  $\epsilon\geq \sqrt{n /m}$, outputs with high probability the explicit
  description of an $\epsilon$-spectral sparsifier of $G$ with
  $\widetilde O\parens {n /\epsilon^2}$ edges, using
  $\widetilde O\parens{\sqrt{mn}/\epsilon}$ queries to $\cO_G$ and in time
  $\widetilde O\parens{\sqrt{mn}/\epsilon}$.
\end{theorem}

The next quantum algorithm proposed by~\citet{hamoudi2022preparing} provides an
efficient approach to prepare many copies of a quantum state.

\begin{theorem}[Preparing Many Copies of a Quantum State, Theorem 1 in~\citet{hamoudi2022preparing}]\label{thm:quantum-prob-sample}
  There exists a quantum algorithm that, given oracle access $\cO_{w}$ to a
  vector $w \in \mathbb{R}_{\geq 0}^n$ (0-indexed)
  ($\cO_w:\ket{i}\ket{0}\mapsto \ket{i}\ket{w_i},\forall i \in \sqb{n}$), and
  $ k\in [n]$, with high probability, outputs $k$ copies of the state $\ket{w}$,
  where 
  \[
  \ket{w}=\frac{1}{\sqrt{W}}\sum_{i\in \sqb{n}_0}\sqrt{w_i} \ket{i}
  \]
  with $W={\sum_{i\in \sqb{n}_0}w_i}$.
  The algorithm uses $\widetilde{O} (\sqrt{nk})$ queries to $\cO_{w}$, and runs
  in $\widetilde{O} (\sqrt{nk})$ time.
\end{theorem}

By performing measurements on each of the generated quantum state copies, a
sample sequence is produced, where each element $i$ is selected with a
probability proportional to $w_i$.
This leads to the following corollary.

\begin{corollary}\label{cor:quantum-prob-sample}
  There exists a quantum algorithm $\mathsf{MultiSample}(\cO_w, k)$ that, given
  query access $\cO_w$ to a vector $w \in \mathbb{R}_{\geq 0}^n$ and integer
  $k\in [n]$, outputs with high probability a sample sequence
  $\sigma\in \sqb{n}^k$ such that each element $i$ is sampled with probability
  proportional to $w_i$.
  The algorithm uses $\widetilde O\parens {\sqrt{nk}}$ queries to $\cO_w$, and
  runs in $\widetilde O\parens {\sqrt{nk}}$ time.
\end{corollary}

In addition to the quantum algorithms mentioned above, we also require quantum
sum estimation, which provides a quadratic speedup over classical approaches.

\begin{theorem}[Quantum Sum Estimation, Lemma 3.1 in~\citet{li2019sublinear}]\label{thm:quantum-sum-estimate}
  There exists a quantum algorithm $\mathsf{SumEstimate}\parens{\cO_w,\epsilon}$
  that, given query access $\cO_w$ to a vector $w \in \mathbb{R}^n_{\geq 0}$ and
  $\epsilon >0$, outputs with high probability an estimate $\widetilde s$ for
  $s=\sum_{i\in \sqb{n}}w_i$ satisfying $\abss{\widetilde s -s}\leq \epsilon s$,
  using $\widetilde O\parens{\sqrt{n}/\epsilon}$ queries to $\cO_w$ and in
  $\widetilde O\parens{\sqrt{n}/\epsilon}$ time.
\end{theorem}

\section{Quantum Algorithm for Leverage Score Overestimates}

In this section, we will introduce the notion of hyperedge leverage scores,
which is a generalization of leverage scores of edges in a graph.
We then define the concept of leverage score overestimates, which are one-side
bounded estimates of hyperedge leverage scores.
Finally, we propose a quantum algorithm that computes the overestimates given
query access to a hypergraph.

To define hyperedge leverage scores, we first introduce the concept of
underlying graphs.
\begin{definition}[Underlying Graph]\label{def:underlying-graph}
  Given an undirected weighted hypergraph $H =\parens{V, E, w}$, an underlying
  graph of $H$ is defined as a multigraph $G = \parens{V, F, c}$ with edge set
  $F =\set{\parens{e,f}: f \in \binom{e}{2}, e \in E} $ and weights
  $c \in \mathbb{R}_{\geq 0}^{F}$, satisfying
  \begin{equation}\label{eq:underlying-graph-constraints}
    w_e=\sum_{{f}\in \binom {e}{2}}c_{e,f}\ , \quad \quad \forall e \in E .
  \end{equation}
\end{definition}
Note that if the hypergraph contains $m$ hyperedges with rank $r$, its
underlying graph can have up to $ mr(r-1)/2$ edges.
The multiple edges in the underlying graph are labeled according to the
hyperedges they originate from.

With this concept, we define the hyperedge leverage score as follows: Given a
hypergraph $H$ and one of its corresponding underlying graphs $G$, the
\emph{leverage score of a hyperedge} $e \in E$, is defined as
\begin{equation}
\ell_e :=w_e R_e, 
\end{equation}
where $R_e =\max\sets {R_f: \forall f \in \binom{e}{2}}$, 
and $R_f$ represents the effective resistance of the edge $f$ in the underlying
graph $G$.

We remark that, the choice of the underlying graph $G$ will greatly influence
the hyperedge leverage scores.
For our sparsification purpose, we want to bound the total sum of hyperedge
leverage scores by $O(n)$, which determines the the size of the resulting
hypergraph sparsifiers.
Therefore, we define the following notion of hyperedge leverage score
overestimates, which are entry-wise upper bounds for leverage scores with a
specific underlying graph, such that the total sum is bounded by a parameter
$\nu = O(n)$.

\begin{definition}[Hyperedge Leverage Score Overestimate, adapted from~{\citet[Definition 1.3]{jambulapati2023chaining}}]\label{def:hy3peredge-overestimate}
  Given a hypergraph $H =\paren{V,E,w}$, we say
  $z \in \mathbb{R}^{E}_{\geq 0}$ is a $\nu$-(bounded hyperedge leverage
  score) overestimate for $H$ if $\norm{z}_1\leq \nu$ and there exists a
  corresponding underlying graph $G =\paren{V,F,c}$, satisfying the constraints
  \cref{eq:underlying-graph-constraints}, such that for all $e \in E$,
  $ z_e \geq \ell_e $.
\end{definition}

To obtain an overestimate, we need to handle the weights of the underlying
graph, which contains $O\parens{mr^2}$ edges.
Nevertheless, it is sufficient to manage only $O\parens{mr}$ edges by replacing
each hyperedge with a sparse subgraph (e.g., a star graph with up to $r-1$
edges) rather than a complete clique (\cref{def:underlying-graph}).
We refer to the resulting graph as a sparse underlying graph.
\begin{definition}[Sparse Underlying Graph]\label{def:sparse-underlying-graph}
  Given an undirected weighted hypergraph $H =\parens{V, E, w}$, a sparse
  underlying graph of $H$ is defined as $G = \parens{V, F, c}$ with edge set
  $ F =\set{\parens{e,f}: f \in S_e , e \in E} $ and weights
  $c \in \mathbb{R}_{\geq 0}^{F}$, satisfying
  \begin{equation}\label{eq:sparse-underlying-graph-constraints}
  w_e  =\sum_{{f}\in S_e}c_{e,f}, \quad \quad \forall e \in E.
  \end{equation}
  For each $e\in E$, we fix an arbitrary vertex $a_e \in e$ and define
  $S_e =\sets{f \in \binom{e}{2} : a_e\in f}$.
\end{definition}

Now we present our quantum algorithm of hyperedge leverage score overestimates,
which is a key step for our main quantum algorithm for hypergraph
sparsification.
Our algorithm is inspired by the approximating John ellipsoid algorithm proposed
in~\citet{cohen2019near} and the group leverage score overestimate algorithm
in~\citet{jambulapati2023chaining}.
The input of our algorithm includes a quantum oracle to the hypergraph, the
number of iterations $T$, the graph sparsification parameter $\alpha_1$, and the
effective resistance approximation factor $\alpha_2$.
The output of our quantum algorithm is a data structure, which could provide a
query access to the hyperedge leverage score overestimates
(see~\cref{prop:overestimate-preparation} for a formal description of the
output).

Recall that the choice of the underlying graph $G$ determines the hyperedge
leverage scores, as well as their overestimates.
Our algorithm iteratively adjusts the edge weights of the underlying graph over
roughly $\log r$ rounds to construct a suitable $G$.
In each iteration, we use quantum graph sparsification to reduce the graph’s
size and reassign hyperedge weights to the edges of the underlying graph
according to edge leverage scores.
To maintain overall efficiency, this process is implemented entirely quantumly
through a series of quantum subroutines.
The process begins with $\mathsf{WeightInitialize}$, which is employed during
the first iteration to establish quantum query access to the weights of the
underlying graph via queries to the original hypergraph
(\cref{prop:weight-init}).
In each iteration, the system utilizes $\mathsf{UGraphStore}$ to provide quantum
query access to the stored weights of the sparsifier obtained from
$\mathsf{GraphSparsify}$ (\cref{prop:quantum-underlying-graph-store}).
Subsequently, $\mathsf{EffectiveResistance}$ enables efficient quantum queries
to the approximate effective resistance of a graph
(\cref{prop:quantum-effective-resistance-oracle}).
For the next iteration, $\mathsf{WeightCompute}$ implements quantum query access
to the updated weights of the sparse underlying graph
(\cref{prop:weight-compute}).
The complete algorithm is presented in \cref{alg:QHLSO}.

\begin{algorithm}[htb]
  \caption{Quantum Hyperedge Leverage Score Overestimates
  $\mathsf{QHLSO}(\cO_H, T,\alpha_1,\alpha_2)$}\label{alg:QHLSO}
  \begin{algorithmic}[1] 
  \REQUIRE{} Quantum Oracle $\cO_H$ to a hypergraph $H=(V,E,w )$ with
  $\abss{V}=n,\abss{E}=m$, rank $r$; the number of episodes $T\in \mathbb{N}$;
  positive real numbers $ \alpha_1,\alpha_2 \in \mathbb{R}$.
  \ENSURE{} An instance $\cZ$ of $\mathsf{QOverestimate}$ which stores the
  vector $z$ being an $O(n )$-overestimate for $H$.
  \STATE{} Let $U _{G (1)} = \mathsf{WeightInitialize}(\cO_H )$.
  \FOR{$t=1$ to $T$} %
  \STATE{}
  $\widetilde G^{(t)}=\parens{V,\widetilde F^{(t)},\widetilde c^{(t)}}
  \gets \mathsf{GraphSparsify}(U_{G (t)}, \alpha_1)$.
  \STATE{} $\cG^{(t)} \gets \mathsf{UGraphStore}\parens{\widetilde G^{(t )}}$.
  \STATE{}
  $\cR^{(t)} \gets \mathsf{EffectiveResistance}(\widetilde G^{(t)}, \alpha_2) $.
  \STATE{}
  $U_{G(t+1)}= \mathsf{WeightCompute} \parens{\cO_H, \cR^{(t)}, \cG^{(t)}}$.
  \ENDFOR{} %
  \STATE{}
  $C_1\gets 2\parens{1+ \frac{\alpha_1+\alpha_2}{1-\alpha_1}}
  \cdot \exp\paren{{\log r} /T} $.
  \STATE{}
  $\cZ \gets \mathsf{QOverestimate}\parens{\sets{\cG^{(t)}: t \in \sqb{T}},
    \sets{\cR^{(t)} : t \in \sqb{T}},\mathcal{O}_H , C_1,T}$.
  \end{algorithmic}
\end{algorithm}

The algorithm's complexity is described in the following theorem.

\begin{theorem}[Quantum Hyperedge Leverage Score Overestimates]\label{thm:overestimate-quantum}
  There exists a quantum algorithm
  $\mathsf{QHLSO}\parens{\cO_H, T ,\alpha_1,\alpha_2}$ that, given integer
  $T=O\parens{\log r}$, positive real numbers $\alpha_1,\alpha_2 < 1$ and query
  access $\cO_{H}$ to a hypergraph $H =\paren{V,E,w}$ with
  $\abss{E}=m, \abss{V}=n, w \in \mathbb{R}^E_{\geq 0}$, and rank $r$, the
  algorithm runs in time $\widetilde O\parens{r \sqrt{mnr}}$.
  Then, with high probability, it provides query access to a $\nu$-overestimate
  $z$ with $\nu=O\parens{n} $, where each query to $z$ requires
  $\widetilde O\parens{r}$ time.
\end{theorem}

Due to space constraints, the proof of \cref{thm:overestimate-quantum} is
deferred to \cref{sec:proof-of-overestimate}.

\section{Quantum Hypergraph Sparsification}

Assuming query access to a $\nu$-overestimate, we aim to implement the sampling
scheme in a quantum framework.
By leveraging \cref{cor:quantum-prob-sample}, we can sample a sequence
$\sigma=\parens{\sigma_i:\sigma_i \in E}$ where each element $e$ is sampled with
probability proportional to $z_e$.
By combining the information of each sampled $e$ with the normalization factor
obtained via $\mathsf{SumEstimate}$, we assign appropriate weights to the
sampled edges to construct the final sparsifier.
The complete algorithm is outlined in \cref{alg:hyper-sparse}.

\begin{algorithm}[htb]
  \caption{Quantum Hypergraph Sparsification
  $\mathsf{QHypergraphSparse}\parens{\cO_H, \epsilon}$}\label{alg:hyper-sparse}
  \begin{algorithmic}[1]
  \REQUIRE{} Quantum Oracle $\cO_H$ to a hypergraph $H=(V,E,w )$ with
  $\abss{V}=n,\abss{E}=m$, rank $r$; accuracy $ \epsilon >0$.
  \ENSURE{} An $\epsilon$-spectral sparsifier of $H$, denoted by
  $\widetilde H=\parens{V, \widetilde E, \widetilde w}$,
  $\abs{\widetilde E}=O\parens{n\log n \log r/\epsilon^2}$.
  \STATE{}
  $\widetilde E =\emptyset , \widetilde w =0,
  M \leftarrow \Theta \paren{n \log n\log r /\epsilon^2}$.
  \STATE{} $\cZ \gets \mathsf{QHLSO}(\cO_H, {\log \parens{r-1}}, 0.1, 0.1)$.
  \STATE{} $\sigma\gets \mathsf{MultiSample}(\cZ.\mathsf{Query}, M)$.
  \STATE{} $s \gets \mathsf{SumEstimate}(\cZ.\mathsf{Query}, 0.1)$.
  \FOR{$i=1$ to $M$} \STATE{} $w_{\sigma_i}\gets$ measurement outcome of the
  second register of $\cO_H^{\textup{wt}}\ket{\sigma_{i}}\ket{0}$.
  \STATE{} $z_{\sigma_i}\gets$ measurement outcome of the second register of
  $\cZ.\mathsf{Query}\ket{\sigma_{i}}\ket{0}$.
  \STATE{}
  $\widetilde E \leftarrow \widetilde E \cup \sets{\sigma_i},
  \widetilde w_{\sigma_i} \leftarrow \widetilde w_{\sigma_i} +
  w_{\sigma_i} \cdot s/\parens{M z_{\sigma_i}} $.
  \ENDFOR{}
  \end{algorithmic}
\end{algorithm}

\begin{theorem}[Quantum Hypergraph Sparsification]\label{thm:quantum-hypergraph-sparsification}
  There exists a quantum algorithm that, given query access to a hypergraph
  $H =\parens{V,E,w}$ with $\abs{E}=m, \abs{V}=n, w \in \mathbb{R}^E_{\geq 0}$,
  rank $r$, and $\epsilon> 0$, outputs with high probability the explicit
  description of an $\epsilon$-spectral sparsifier of $H$ with
  $ O\parens{n\log n \log r /\epsilon^2}$ hyperedges, in time
  $\widetilde O\parens{r\sqrt{mnr} +r\sqrt{mn}/\epsilon}$.
\end{theorem}
The proof of correctness for the algorithm follows closely the chaining argument
in~\citet{lee2023spectral} and~\citet{jambulapati2023chaining}.
Due to space constraints, we provide the detailed proof in
\cref{sec:proof-of-hypergraph-sparsification}.

\begin{remark}\label{re:main-theorem-1}
  We assume $\epsilon \geq \sqrt{n/m}$, as sparsification is beneficial only
  when the sparsifier contains less $m$ hyperedges.
  It is also generally the case that $m \geq nr$, as the objects being
  sparsified are dense.
  It's worth noting that the time complexity
  $\widetilde O\paren{r \sqrt{mnr}+ r\sqrt{mn}/\epsilon}$ simplifies to
  $\widetilde O \paren{r\sqrt{mn}/\epsilon}$ whenever $\epsilon \geq \sqrt{n/m}$
  and $m \geq nr$.
\end{remark}

\begin{remark}\label{re:main-theorem-2}
  For the hypergraph with constant rank $r$, the above complexity contrasts with
  the classical lower bound of $\Omega(m)$.
  This lower bound arises from the fact that, in the case of graphs ($r = 2$),
  there exists an $\Omega\paren{m}$ classical query lower bound for determining
  whether a graph is connected, which establishes the same lower bound for both
  cut sparsifiers and spectral sparsifiers.
\end{remark}

\section{Applications}\label{sec:applications}

As a direct corollary of \cref{thm:quantum-hypergraph-sparsification}, we can
compute a cut sparsifier for a hypergraph in sublinear time.
\begin{corollary}[Quantum Hypergraph Cut Sparsification]\label{cor:quantum-hypergraph-cut}
  There exists a quantum algorithm that, given query access to a hypergraph
  $H =\parens{V,E,w}$ with $\abs{E}=m, \abs{V}=n,w \in \mathbb{R}^E_{\geq 0}$,
  rank $r$, and $\epsilon> 0$, outputs with high probability the explicit
  description of an $\epsilon$-cut sparsifier of $H$ with
  $ O\parens{n\log n \log r /\epsilon^2}$ hyperedges in time
  $ \widetilde O\parens{r\sqrt{mnr} +r\sqrt{mn}/\epsilon}$.
\end{corollary}
Similar to the case for graphs, quantum hypergraph cut sparsification
facilitates faster approximation algorithms for cut problems.
Below, we highlight two such applications.

\paragraph{Mincut} Given a hypergraph $H=\parens{V,E,w}$, the hypergraph mincut
problem asks for a vertex set $S:\emptyset\subsetneq S\subsetneq V$ that
minimizes the energy $Q_H \parens{S}$.
To the best of our knowledge, the fastest algorithm for computing the mincut in
a hypergraph without error runs in $\widetilde O\parens{mnr}$
time~\citep{klimmek1996,mak2000fast}.
By applying \cref{cor:quantum-hypergraph-cut}, we first sparsify the hypergraph
and then apply the mincut algorithm to the cut sparsifier.
This gives a quantum algorithm that, with high probability, outputs a
$\parens{1+\epsilon}$-approximate mincut in time
$\widetilde O \parens{r\sqrt{mn}/\epsilon+rn^2}$, which is sublinear in the
number of hyperedges.

\begin{corollary}[Quantum Hypergraph Mincut Solver]
  There exits a quantum algorithm that, given query access to a hypergraph
  $H =\parens {V,E,w}$ with $\abs{E}=m, \abs{V}=n,w \in \mathbb{R}^E_{\geq 0}$,
  rank $r$, and $\epsilon> 0$, outputs with high probability the
  $\parens{1+\epsilon}$-approximate mincut of $H$ in time
  $ \widetilde O\parens{r\sqrt{mnr} +r\sqrt{mn}/\epsilon+rn^2}$.
\end{corollary}

\paragraph{$s$-$t$ mincut}Given a hypergraph $H=\parens{V,E,w}$ and two vertices
$s,t\in V$, the $s$-$t$ mincut problem seeks a vertex set $S\subseteq V$ with
$\abs{S\cap\sets{s,t}}=1$(i.e., either $s\in S$ or $t\in S$) that minimizes the
energy $Q_H\parens{S}$.
The standard approach for computing an $s$-$t$ mincut in a hypergraph is
computing the $s$-$t$ maximum flow in an associated digraph with $O\parens{n+m}$
vertices and $O\parens{mr}$ edges~\citep{Lawler73}.
And an $s$-$t$ maximum flow in such graph can be found in
$\widetilde O\parens{mr \sqrt{m+n}}$ time~\citep{LeeS14}.
By combining \cref{cor:quantum-hypergraph-cut} with the aforementioned approach,
we can compute a $\parens{1+\epsilon}$-approximate $s$-$t$ mincut in time
$\widetilde O\parens{r \sqrt{mn}/\epsilon}$ whenever $m=\Omega\parens{n^2}$,
which is sublinear in number of hyperedges.

\begin{corollary}[Quantum Hypergraph $s$-$t$ Mincut Solver]
  There exits a quantum algorithm that, given query access to a hypergraph
  $H =\parens {V,E,w}$ with $\abs{E}=m, \abs{V}=n,w \in \mathbb{R}^E_{\geq 0}$,
  rank $r$, two vertices $s,t\in V$, and $\epsilon> 0$, outputs with high
  probability the $\parens{1+\epsilon}$-approximate $s$-$t$ mincut of $H$ in
  time $ \widetilde O\parens{r\sqrt{mnr} +r\sqrt{mn}/\epsilon+rn^{3/2}}$.
\end{corollary}

\section{Conclusion and Future works}
In this paper, we present a quantum algorithm for hypergraph sparsification with
time complexity $\widetilde O\parens{r\sqrt{mn}/\epsilon+r\sqrt{mnr}}$, where
$m,n,r,\epsilon$ are the number of hyperedges, the number of vertices, rank of
hypergraph and precision of sparsifier, respectively.

Our paper naturally raises several open questions for future work. For
instance:

\begin{itemize}
  \item Quantum graph sparsification has directly led to the development of
    numerous quantum algorithms, including max-cut for
    graphs~\citep{apers2022quantum}, graph minimum cut
    finding~\citep{apers2020quantum}, graph minimum $s$-$t$ cut
    finding~\citep{apers2021sublinear}, motif
    clustering~\citep{cade2023quantum}.
    A natural question arises: for more related problems in hypergraphs,
    such as hypergraph-$k$-cut, minmax-hypergraph-$k$-partition, can we
    design faster quantum algorithms compared to classical ones?

  \item We conjecture that the runtime of our quantum algorithm is tight up to
    polylogarithmic factors when $\epsilon \leq 1/\sqrt{r}$.
    Can we establish a quantum lower bound of
    $\Omega \parens{r\sqrt{mn}/\epsilon}$ for hypergraph sparsification?
    Alternatively, can we improve the time complexity for quantum hypergraph
    sparsification, or further enhance the runtime for hypergraph cut
    sparsification?
    Furthermore, state-of-the-art hypergraph cut sparsification achieves a
    size of $O\parens{n\log n /\epsilon^2}$ (without the $\log r$ factor) in
    $\widetilde O\parens {mn+n^{10}/\epsilon^7}$
    time~\citep{chen2020near}---can we design a faster quantum algorithm
    that matches this size?

  \item Hypergraph sparsification has been extended to various frameworks,
    including online hypergraph sparsification~\citep{soma2024online},
    directed hypergraph sparsification~\citep{oko2023nearly}, generalized
    linear models sparsification~\citep{jambulapati2024sparsifying}, and
    quotient sparsification for submodular
    functions~\citep{quanrud2024quotient}.
    A natural question is can we develop specialized quantum
    algorithms for these settings?
\end{itemize}



\section*{Impact Statement}
This paper presents work whose goal is to advance the field of 
Machine Learning. There are many potential societal consequences 
of our work, none of which we feel must be specifically highlighted here.

\section*{Acknowledgements}
The work was supported by National Key Research and Development
Program of China (Grant No.\ 2023YFA1009403), National Natural Science
Foundation of China (Grant No.\ 12347104), Beijing Natural Science Foundation
(Grant No.\ Z220002), and Tsinghua University.


\bibliography{main}
\bibliographystyle{icml2025}

\newpage
\appendix
\onecolumn
%

\section{Useful properties of effective resistance}\label{sec:properties-of-effective-resistance}

It's the well-known fact in graph theory that the effective resistance defines a
metric on the vertices of a graph.
Below, we outline several key properties of effective resistance that will be
utilized in this work.

\begin{lemma}[Foster]\label{le:Foster}
  For a weighted graph $G=\parens{V,F,c}$, let $R_{ij}$ represent the effective
  resistance between vertices $i$ and $j$.
  Then, it holds that $\sum_{\sets{i,j} \in F} c_{ij} R_{ij} \leq n $.
\end{lemma}

\begin{proof}For a edge $\sets{i,j}$, recall that
  $R_{ij}=\paren{\delta_i -\delta_j}^\top L_G^{+}\paren{\delta_i -\delta_j}$,
  then
  \begin{equation*}
  \sum_{\sets{i,j} \in F} c_{ij} R_{ij} =\sum_{\sets{{i,j}}\in F} \trace \paren{c_{ij}\paren{\delta_i -\delta_j}\paren{\delta_i -\delta_j}^\top L_G^{+}}=\trace\paren{L_G L_G^{+}}\leq n-1
  \end{equation*}
  since rank of $L_G$ is at most $n-1$.
\end{proof}

\begin{lemma}[Convexity, Lemma 3.4 in~\citet{cohen2019near}]\label{le:effective-resistance-convexity}
  For a weighted graph $G=\parens{V,F,c}$, the function $ \log R_f\parens{c}$ is
  convex with respect to $c$.
\end{lemma}

\section{Proof of \cref{thm:overestimate-quantum}}\label{sec:proof-of-overestimate}

For a hyperedge $e$, let $S_e $ represent the set
$\sets{f \in \binom{e}{2} : a\in f}$, where $a$ is a fixed vertex in $e$.
\begin{proposition}[Weight Initialization for Overestimates]\label{prop:weight-init}
  Suppose $H = \parens{V,E,w}$ is a hypergraph with vertex set $V$ of size $n$,
  edge set $E$ of size $m$, weight function $w: E\to \mathbb{R}_{\geq 0}$, and
  $\mathcal{O}_H $ is a quantum oracle to $H$.
  Then, there exists a quantum algorithm
  $\mathsf{WeightInitialize}(\mathcal{O}_H)$, that satisfies
  \[
  \mathsf{WeightInitialize}(\mathcal{O}_H)\ket{e}\ket{f}\ket{0} = \ket{e}\ket{f}\kets{c_{e,f}^{(1)}}
  \]
  with $ c_{e,f}^{(1)}=w_e /\parens{\abss{e}-1}$, for $\forall e \in E$ and
  $\forall f =\sets{i,j}\in S_e$, and performs no action if
  $f=\sets{i,j}\notin S_e$, using $O(1)$ queries to $\mathcal{O}_H$, in
  $\widetilde O(1)$ time.
  Here, we represent $\ket{f}$ as $\ket{i}\ket{j}$ for $f =\sets{i,j}$.
\end{proposition}

\begin{proof}
  Consider the following procedure: for any $e\in E$ and $f\in S_e$, we have
  \[
  \begin{aligned}
    \ket{e} \ket{f} \ket{0} \ket{0} \ket{0}
    & \xmapsto{\mathcal{O}_H^{\textup{wt}},\mathcal{O}_H^{\textup{size}}}
    \ket{e} \ket{f} \ket{w_e} \ket{\vert e\vert} \ket{0} \\
    & \xmapsto{U_{\textup{fwd}}}
    \ket{e} \ket{f} \ket{w_e} \ket{\abss{e}-1} \ket{0} \\
    & \xmapsto{U_{\textup{div}}}
    \ket{e} \ket{f} \ket{w_e} \ket{\abss{e}-1}  \kets{c_{e,f}^{(1)}} \\
    & \xmapsto{U_{\textup{fwd}}^{\dagger}}
    \ket{e} \ket{f} \ket{w_e} \ket{\abss{e}}  \kets{c_{e,f}^{(1)}} \\
    & \xmapsto{{\mathcal{O}_H^{\textup{size}}}^{\dagger},{\mathcal{O}_H^{\textup{wt}}}^{\dagger}}
    \ket{e} \ket{f} \ket{0} \ket{0}  \kets{c_{e,f}^{(1)}},
  \end{aligned}
  \]
  where $U_{\textup{fwd}}$ satisfies $U_{\textup{fwd}}\ket{i} = \ket{i-1}$,
  $U_{\textup{div}}$ satisfies
  $U_{\textup{div}}\ket{i}\ket{j}\ket{0}=\ket{i}\ket{j}\ket{i/j}$.
  This procedure uses $4$ queries to $\mathcal{O}_H$ and $\widetilde O(1)$
  additional arithmetic operations.
  Therefore, let
  $\mathsf{WeightInitialize}(\mathcal{O}_H) =
  {\cO_{H}^{\textup{wt}}}^\dagger{\cO_{H}^{\textup{size}}}^\dagger U_{\textup{fwd}}^\dagger
  U_{\textup{div}} U_{\textup{fwd}}\cO _H^{\textup{size}}\cO_{H}^{\textup{wt}}$,
  we know it satisfies the requirement stated in the proposition.
\end{proof}

Recall the following classical algorithm for efficiently computing the effective
resistances of a graph proposed by~\citet{spielman2011graph}.

\begin{theorem}[Effective Resistance Oracle, Theorem 2 in~\citet{spielman2011graph}]\label{thm:approximate-resistance-oracle}
  There exists a (classical) algorithm
  $\mathsf{ClassicalEffectiveResistance}(G, \varepsilon)$ such that for any
  $\varepsilon > 0$ and graph $G= \parens{V, F, c}$ with vertex set $V$ of size
  $n$, edge set $F$ of size $m$, and weights $c\in \mathbb{R}_{\ge 0}^F$, with
  high probability, returns a matrix $Z_G$ of size $p \times n $ with
  $p= \ceils{24\log n /\epsilon^2}$ satisfying
  \begin{equation*}
  \parens{1-\epsilon}R_{ab}\leq \norm{Z_G \parens{\delta_a -\delta_b}}^2 \leq
  \parens{1+\epsilon}R_ {ab}
  \end{equation*}
  for every pair of $a, b\in V$, where $R_{ab}$ is the effective resistance
  between $a$ and $b$, in $\widetilde O\parens{m /\epsilon^2}$ time.
\end{theorem}

We require a quantum variant of this effective resistance computation, as
outlined below.

\begin{proposition}[Quantum Effective Resistance Oracle, {\citet[Claim 7.9]{apers2022quantum}}]\label{prop:quantum-effective-resistance-oracle}
  Let $G= \parens{V, F, c}$ be a graph with a vertex set $V$ of size $n$, an
  edge set $F$ of size $m$, and weights $c: \in \mathbb{R}_{\ge 0}^F$.
  For $\varepsilon > 0$, there is a quantum data structure
  $\mathsf{EffectiveResistance}$, that supports the following operations:
  \begin{itemize}
  \item Initialization: $\mathsf{EffectiveResistance}(G, \varepsilon)$,
      outputs an instance $\mathcal{R}$, in
      $\widetilde O\parens{m /\epsilon^2}$ time.
  \item Query: $\mathcal{R}.\mathsf{Query}$, outputs a unitary satisfying
      \[
      \mathcal{R}.\mathsf{Query} \ket{a} \ket{b} \ket{0} =
      \ket{a}\ket{b}\vert \widetilde R_{ab}\rangle
      \]
      with
      $(1-\epsilon) R_{ab}\leq \widetilde R_{ab}\leq (1+\epsilon) R_{ab} $
      for every pair of vertices $a,b \in V$, and $R_{ab}$ being the
      effective resistance between vertices $a$ and $b$, in
      $\widetilde O(1/\epsilon^2)$ time.
  \end{itemize}
\end{proposition}

\begin{proof}
  First, we use the algorithm
  $\mathsf{ClassicalEffectiveResistance}\parens{G , \epsilon}$ to obtain the
  matrix $Z_G$ and store each entries of matrix in QRAM\@.
  This allows us to access the matrix through a unitary $U^{\textup{}}_{Z_G}$
  such that
  \begin{equation*}
  U_{Z_G}\ket{i}\ket{j}\ket{0}=\ket{i}\ket{j}\ket{Z_G \parens{i,j}}
  \end{equation*}
  where $Z_G\paren{i,j}$ is the entry in the $i$-th row and $j$-th column of the
  matrix $Z_G$, $i \in \sqb{n}, j \in \sqb{q}$ with
  $q=\ceils{24 \log n /\epsilon^2}$.
  The time of computing and storing $Z_G$ is
  $\widetilde O\parens{m /\epsilon^2}$, and each query of $U_{Z_G}$ has a time
  complexity of $\widetilde O\parens{1}$.
  Consider the following procedure:
  \begin{align*}
  \ket{i} \ket{j} \ket{0} \ket{0} \ket{0}
  & \xmapsto{U_{Z_G}^\prime} \ket{i}\ket{j}\Bigl(\bigotimes_{k=1}^{q}
    \kets{Z_G\paren{i,k}}\Bigr) \Bigl(\bigotimes_{k=1}^{q}
    \kets{Z_G\parens{j,k}}\Bigr) \kets{0}\\
  & \overset{\textup{denote}}{=}\ket{i}\ket{j}\kets{Z_G^{i}} \kets{Z_G^{j}}\ket{0}\\
  & \xmapsto{U_ {\textup{minus}}} \ket{i}\ket{j}\kets{Z_G^{i}}
    \kets{Z_G^{j}}\kets{Z_G^{i}-Z_G^{j}}\\
  & \xmapsto{U_ {\textup{square}}} \ket{i}\ket{j}\kets{Z_G^{i}}
    \kets{Z_G^{j}}\kets{\widetilde R_{ij}}\\
  & \xmapsto{{U_ {Z_G}^\prime}^{\dagger}}
    \ket{i}\ket{j}\kets{0}\kets{0}\kets{\widetilde R_{ij}}
  \end{align*}
  where $Z_G^i$ is the $i$-th row of the matrix $Z_G$, and $U_{Z_G}^\prime$ can
  be implemented using $ O\parens{q}$ queries of $ U_{Z_G}$;
  $U_{\textup{minus}}$ satisfies
  $U_{\textup{minus}}\ket{i}\ket{j}\ket{0} =\ket{i}\ket{j} \ket{i-j}$,
  $U_{\textup{square}}$ satisfies $U_{\textup{square}}\ket{i}\ket{0}=\ket{i}\ket{i^2}$.
  This procedure uses $\widetilde O\parens{1/\epsilon^2}$ queries to $U_{Z_G}$
  and $\widetilde O(1)$ additional arithmetic operations.
  Therefore, let
  $\cR.\mathsf{Query} = {U_{Z_G}^\prime}^\dagger
  U_{\textup{square}} U _{\textup{minus}}U_{Z_G}^\prime$,
  we confirm that this satisfies the requirements outlined in the proposition.
\end{proof}

The following proposition formalizes the procedure of storing a graph in QRAM,
which is made straightforward by the capabilities of QRAM\@.
\begin{proposition}[Quantum Underlying Graph Storage]\label{prop:quantum-underlying-graph-store}
  Let $H = \parens{V,E,w}$ is a hypergraph with vertex set $V$ of size $n$, edge
  set $E$ of size $m$, weights $w\in \mathbb{R}_{\geq 0}^E$.
  Suppose $G= \parens{V, F, c}$ is a sparse underlying graph $G$ of $H$.
  There is a quantum data structure $\mathsf{UGraphStore}$, that supports the
  following operations:
  \begin{itemize}
  \item Initialization: $\mathsf{UGraphStore}(G)$, outputs an instance
      $\mathcal{G}$, in $\widetilde O\parens{mr}$ time.
  \item Query: $\mathcal{G}.\mathsf{Query}$, outputs a unitary satisfying
      \[
      \mathcal{G}.\mathsf{Query} \ket{e}\ket{f} \ket{0} =
      \ket{e}\ket{f}\ket{c_{e,f}}
      \]
      for every edge $\parens{e,f} \in F$, where $f = \sets{i,j}\in S_e$,
      and performs no action if $\parens{e,f} \notin F$, in
      $\widetilde O(1)$ time.
      Here, we represent $\ket{f}$ as $\ket{i}\ket{j}$ for $f =\sets{i,j}$.
  \end{itemize}
\end{proposition}

\begin{proof}
  For the graph $G$ with edge set $F$ of size at most $m \parens{r-1}$, the
  initialization step is to store all the weights $c_{e,f}$ for the edges with
  indices $(e,f)\in F $ into an array using a QRAM of size $\widetilde O(mr)$
  and in $\widetilde O(mr)$ QRAM classical write operations.

  For the query operation, the $\cG.\mathsf{Query}$ is the QRAM quantum query
  operation to the above array.
\end{proof}

The following proposition concerns weight updates in the quantum overestimation
algorithm.
\begin{proposition}[Weight Computation for Overestimates]\label{prop:weight-compute}
  Let $H = \parens{V,E,w}$ be a hypergraph with vertex set $V$ of size $n$, edge
  set $E$ of size $m$, weights $\in \mathbb{R}_{\geq 0}^E$.
  Suppose $\mathcal{O}_H $ is a quantum oracle to $H$, $\cG$ and $\cR$ represent
  the instances of $\mathsf{UGraphStore}$ and $\mathsf{EffectiveResistance}$ for
  the sparse underlying graph $G =\parens{V,F,c}$, respectively.
  Then, there exists a quantum algorithm
  $\mathsf{WeightCompute}(\mathcal{O}_H, \cG, \cR )$, such that
  \[
  \mathsf{WeightCompute}(\mathcal{O}_H, \cG, \cR )\ket{e}\ket{f}\ket{0} =
  \ket{e}\ket{f}\kets{c_{e,f}^\prime}
  \]
  where
  \begin{equation}\label{eq:weight-compute}
    c_{e,f}^\prime = \frac{c_{e,f} R_f} {\sum_{g \in S_e} c_{e,g} R_g}\cdot w_e,
  \end{equation}
  and $R_f$ is the query result of $\cR$ on vertices of $f$.
  The algorithm requires $O\parens{1}$ queries to $\cO_H$, $O(r)$ queries to
  both $\cR$ and $\cG$, and runs in $\widetilde O\parens{r}$ time.
\end{proposition}

\begin{proof}
  Consider the following procedure: for any $e\in E$ and $f\in S_e$, we have
  \[
  \begin{aligned}
    \ket{e} \ket{f} \ket{0}
    & \xmapsto{\cO_H}
    \ket{e} \ket{f} \ket{0}\paren{\otimes_{i \in e} \ket{i}}\ket{0} \ket{w_e} \ket{0} \\
    & \xmapsto{U_\textup{star}}
    \ket{e} \ket{f} \ket{0}\paren {\otimes_{g\in S_e}  \ket{g} \ket{0}} \ket{w_e} \ket{0} \\
    & \xmapsto{\cG.\mathsf{Query}} \ket{e} \ket{f} \ket{c_{e,f}}\ket{0}
    \paren{\otimes_{g \in S_e}  \ket{g}\ket{c_{e,g}}\ket{0}} \ket{w_e}\ket{0} \\
    & \xmapsto{\cR.\mathsf{Query}} \ket{e} \ket{f} \ket{c_{e,f}}\ket{R_f}\ket{0}
    \paren {\otimes_{g \in S_e}\ket{g} \ket{c_{e,g}} \ket{R_g}\ket{0}} \ket{w_e}\ket{0} \\
    & \xmapsto{U_{\textup{mult}}} \ket{e} \ket{f} \ket{c_{e,f}}\ket{R_f}\ket{w_e c_{e,f} R_f}
    \paren {\otimes_{g \in S_e}\ket{g}  \ket{c_{e,g}} \ket{R_g}\ket{c_{e,g} R_g}} \ket{w_e}\ket{0} \\
    & \xmapsto{\cO_H^\dagger,\cG.\mathsf{Query}^\dagger, \cR.\mathsf{Query}^\dagger}
    \ket{e} \ket{f} \ket{w_e c_{e,f}  R_f}\paren {\otimes_{g \in S_e} \ket{c_{e,g} R_g}} \ket{0} \\
    & \xmapsto{U_{\textup{sum}}} \ket{e} \ket{f} \ket{w_e c_{e,f} R_f}
    \paren {\otimes_{g \in S_e} \ket{c_{e,g}  R_g}} \ket{\Delta_e}\ket{0} \\
    & \xmapsto{U_{\textup{div}}}
    \ket{e} \ket{f} \ket{w_e c_{e,f} R_f} \paren {\otimes_{g \in S_e} \ket{c_{e,g}  R_g}}
    \ket{\Delta_e}\ket{c_{e,f}^\prime}
  \end{aligned}
  \]
  where $\Delta_e $ represents the sum $\sum_{g \in S_e}c_{e,g} R_g$,
  $U_{\textup{star}}$ satisfies
  $U_{\textup{star}}\paren{\otimes_{i \in e}\ket{i}}\ket{0}=\otimes_{g \in S_e}\ket{g}\ket{0}$,
  and $U_{\textup{mult}},U_{\textup{sum}},U_{\textup{div}} $ denote basic
  arithmetic operations---multiplication, addition, and division, respectively,
  as previously.
  It's clear that this procedure meets the requirements stated in the
  proposition, since $\abss{S_e}= O\parens{r}$ for $\forall e \in E$.
\end{proof}

\begin{proposition}[Preparation for Overestimates]\label{prop:overestimate-preparation}
  Let $T \in \mathbb{N}$, $ C \in \mathbb{R}$, $\epsilon \in \mathbb{R}$, and
  $H = \parens{V,E,w}$ be a hypergraph with rank $r$.
  Assume $\mathcal{O}_H $ is a quantum oracle to $H$.
  Suppose $\sets{\cG^{(t)}: t \in \sqb{T}} $ represents a sequence of
  instances of $\mathsf{UGraphStore}$ for the sparse underlying graphs
  $G^{(t)}= \parens{V,F^{(t)} ,c^{(t)}}$ of $H$, and
  $\sets{\cR^{(t)} : t \in \sqb{T}}$ represents a sequence of instances of
  $\mathsf{EffectiveResistance}$ for the corresponding underlying graphs
  $G^{(t)}$ and $\epsilon$.
  Then, there is a quantum data structure $\mathsf{QOverestimate}$, that
  supports the following operations:
  \begin{itemize}
  \item Initialization:
        $\mathsf{QOverestimate}\parens{\sets{\cG^{(t)}: t \in \sqb{T}},
          \sets{\cR^{(t)} : t \in \sqb{T}},\mathcal{O}_H ,C,T}$,
      outputs an instance $\mathcal{Z}$ in
      $\widetilde O\parens{\sum_{t \in \sqb{T}} \tau_{t}}$ time, where
      $\tau_{t}$ denotes the needed time of both
      $\mathsf{UGraphStore}\parens{G^{(t)}}$ and
      $\mathsf{EffectiveResistance}\parens{G^{(t)},\epsilon}$.
  \item Query: $\cZ.\mathsf{Query}$, outputs a unitary such that, for every
      $e \in E $
      \begin{equation*}
      \cZ.\mathsf{Query}\ket{e}\ket{0}=\ket{e}\ket{z_e}
      \end{equation*}
      with
      \begin{equation*}
      z_e =  {C\cdot \frac{1}{T}\sum_{t \in \sqb{T}} \sum_{g \in S_e}  \ell_{e,g}^{(t)}}
      \end{equation*}
      where $ \ell_{e,g}^{(t)} = c_{e,g}^{(t)} R_g^{(t)}$, and
      $ {R_g}^{(t)}$ is the query result of $\cR^{(t)}$ on vertices of $g$.
      This query is executed in $\widetilde O\parens{r \sum \iota_t}$ time,
      where $\iota_t$ represents the time required for querying both
      $\mathsf{UGraphStore}$ and $\mathsf{EffectiveResistance}$ for
      $G^{(t)}$.
  \end{itemize}
\end{proposition}

\begin{proof}
  The case of initialization operation is straightforward.
  Aside from initializing for $\mathsf{UGraphStore}$ and
  $\mathsf{EffectiveResistance}$, we store $C,T$ in QRAM in
  $\widetilde O\parens{1}$ time, allowing access through a unitary
  $U_{C,T}:\ket{i}\ket{0}\to \ket{i}\ket{C/T}$.
  For the query operation, we consider the following procedure:
  \begin{align*}
  \ket{e}\ket{0}
  & \xmapsto{\cO_H, U_{\textup{star}}}
    \ket{e}\paren{\otimes_{t=1}^{T}\paren {\otimes_{g\in S_e} \ket{g} \ket{0}}}\ket{0} \\
  & \xmapsto{\cG^{(t)}.\mathsf{Query},\cR^{(t)}.\mathsf{Query}}
    \ket{e}\paren{\otimes_{t=1}^{T}\paren {\otimes_{g\in S_e} \ket{g} \kets{c_{e,g}^{(t)}}
    \kets{R_g^{(t)}}\kets{0}}}\ket{0} \\
  & \xmapsto{U_{\textup{mult}}}
    \ket{e}\paren{\otimes_{t=1}^{T}\paren {\otimes_{g\in S_e} \ket{g} \kets{c_{e,g}^{(t)}}
    \kets{R_g^{(t)}}\kets{\ell_{e,g}^{(t)}}}}\ket{0} \\
  & \xmapsto{U_{\textup{sum}}}
    \ket{e}\paren{\otimes_{t=1}^{T}\paren {\otimes_{g\in S_e}  \ket{g} \kets{c_{e,g}^{(t)}}
    \kets{R_g^{(t)}}\kets{\ell_{e,g}^{(t)}}}}\kets{\sum_{t}\Delta_e^{(t)}}\ket{0} \\
  & \xmapsto{U_{C,T}, U_{\textup{mult}}}
    \ket{e}\paren{\otimes_{t=1}^{T}\paren {\otimes_{g\in S_e}  \ket{g} \kets{c_{e,g}^{(t)}}
    \kets{R_g^{(t)}}\kets{\ell_{e,g}^{(t)}}}}\kets{\sum_{t}\Delta_e^{(t)}}\ket{z_e}
  \end{align*}
  where $\Delta_e^{(t)} $ represents the sum
  $\sum_{g \in S_e}c_{e,g}^{(t)} R_g^{(t)}$, and
  $U_{\textup{clique}}, U_{\textup{mult}},U_{\textup{sum}} $ denote basic
  arithmetic operations of clique generation, multiplication, and addition,
  respectively, as previously.
  The procedure can be executed in $O\parens{r \sum_{t \in \sqb{T}}\iota_t}$
  time, since $\abss{S_e}= O\parens{r}$ for $\forall e \in E$.
\end{proof}

\begin{proposition}\label{prop:claim-QHLSO}
  Let $\cZ$ be the output of
  $\mathsf{QHLSO}\parens{\cO_H,T, \alpha_1, \alpha_2}$ (\cref{alg:QHLSO}).
  Then, the vector $z$ stored in $\cZ$ is a $\nu$-overestimate for $H$, where
  $\nu=\parens{1+\alpha_2}C_1n$ and $C_1$ is
  determined by $ \alpha_1, \alpha_2,r,T$, as defined in line 8 of
  \cref{alg:QHLSO}.
\end{proposition}

\begin{proof}
  In the algorithm, $\widetilde G^{(t)}$ is a $\alpha_{1}$-spectral sparsifier
  of $G^{(t)}$.
  Let $R^{(\tilde t)}_{f}$ and $R^{(t)}_{f}$ be the effective resistances of
  $\widetilde G^{(t)}$ and $G^{(t)}$ respectively.
  Since effective resistances correspond to quadratic forms in the
  pseudo-inverse of the Laplacian, we have
  $\frac{1}{1+\alpha_1} R_f^{(t)} \leq R^{(\tilde t)}_f \leq
  \frac{1}{1-\alpha_1} R_f^{(t)} , \forall f \in F$.
  According to \cref{prop:quantum-effective-resistance-oracle}, we know that
  $\paren{1-\alpha_2}R_f^{(\tilde t)} \leq \widetilde R_f^{(t)} \leq
  \paren{1+\alpha_2}R_f^{(\tilde t)} ,\forall f \in F$.
  Combining two inequalities we obtain
  \begin{equation*}
  \frac{1-\alpha_2}{1+\alpha_1}\cdot R_f^{(t)} \leq \widetilde R_f^{(t)}
  \leq \frac{1+\alpha_2}{1- \alpha_1}\cdot R_f^{(t)},\quad \forall f\in F.
  \end{equation*}
  Let
  $\alpha_3:=\max \bigl\{1- \frac{1-\alpha_2}{1+\alpha_1},
  \frac{1+\alpha_2}{1- \alpha_1}-1 \bigr\} = \frac{\alpha_1+\alpha_2}{1-\alpha_1}$,
  then $\widetilde R_f^{(t)} $ is an $\alpha_3$-approximate of $R_f^{(t)}$ for
  $\forall f \in F$.
  We define
  $\widetilde \ell_{e,f}^{(t)}= \widetilde c_{e,f}^{(t)} \widetilde R_f^{(t)} $.

  We will show that $z$ is a $\nu$-overestimate with corresponding underlying
  graph $G =\paren{V,F , \bar c}$, where
  $\bar c = \frac{1}{T}\sum_{t\in \sqb{T}} \widetilde c^{(t)} $.
  Specifically, we need to verify the following two conditions:
  \begin{enumerate}
  \item $\norm{z}_1 \leq \nu$,
  \item $z_e \geq w_e R_{e}$ for all $e \in E$ where
      $R_e=\max\sets{R_f: f \in \binom{e}{2}}$.
  \end{enumerate}

  We show the first condition first.
  As $\widetilde \ell_{e,f}^{(t)} = \widetilde c_{e,f}^{(t)} \widetilde
  R_{f}^{(t)} \le (1+\alpha_{2}) \widetilde c_{e,f}^{(t)} R_{f}^{(\tilde t)}$, we have
  \begin{equation*}
  \norm{z}_1 =\sum_{e \in E}  C_1\frac{1}{T}\sum_{t \in \sqb{T}}
  \sum_{g \in S_e}\widetilde \ell_{e,g}^{(t)}
  = C_1\cdot \frac{1}{T}\sum_{t \in \sqb{T}}
  \paren{\sum_{e \in E}\sum_{g \in S_e}\widetilde \ell_{e,g}^{(t)}}
  \leq C_1 \parens{1+\alpha_2} n
  \end{equation*}
  where the final inequality is derived from \cref{le:Foster}.

  We now prove that the second condition also holds.
  For any $e \in E $ we fix $a \in e$.
  Since effective resistance is a metric on vertices, for any $u,v \in e $, it
  follows that
  \begin{equation*}
  R_{u v}\leq R_{u  a}+ R_{av}.
  \end{equation*}
  Consequently, at least one of the two terms on the RHS must be at least
  $R_{u v}/2$.
  By taking the maximum on both sides, we obtain
  \begin{equation*}
  \max_{u,v \in e} R_{uv}\leq 2 \max_{u \in e} R_{au}
  \end{equation*}
  Thus, for any $ e \in E$, we have
  \begin{align*}
  \log \paren{w_e R_e \parens{\overline c}}
  & \leq \log \paren{w_e \cdot 2\max\sets{{R_f(\overline c)}:f \in S_e}} \\
  & \overset{\textup{denote}}{=}\log \parens{w_e \cdot 2 R_{f^{\star}}\parens{\overline c}} \\
  & \overset{(a)}{\leq} \frac{1}{T}\sum_{t \in \sqb{T}} \log \paren{2w_e R_{f^\star} \paren{c^{(t)}}} \\
  & \leq \frac{1}{T}\sum_{t \in \sqb{T}} \log \paren{2 w_e\paren{1+\alpha_3}
    \widetilde R_{f^\star} ^{(t)}}=\frac{1}{T}\sum_{t \in \sqb{T}} \log \paren{2\paren{1+\alpha_3}w_e
    \cdot  \widetilde \ell_{e,f^\star}^{(t)} /\widetilde c_{e,f^\star}^{(t)}} \\
    & \overset{(b)}{=}\frac{1}{T}\sum_{t \in \sqb{T}}
    \paren{\log \paren{\frac{\widetilde c_{e,f^\star}^{(t+1)}\cdot  \sum_{g \in S_e}\widetilde \ell_{e,g}^{(t)}}{ \widetilde c_{e,f^\star}^{(t)}}}}+\log \parens {2\parens{1+\alpha_3}} \\
  & =\frac{1}{T}\sum_{t \in \sqb{T}}
    \paren{\log \paren{\frac{\widetilde c_{e,f^\star}^{(t+1)}}{ \widetilde c_{e,f^\star}^{(t)}}} +
    \log\paren{\sum_{g \in S_e}\widetilde \ell_{e,g}^{(t)}}}+\log \parens {2\parens{1+\alpha_3}} \\
  & \overset{(c)} {\leq} \frac{1}{T}
    \log\paren{ \frac{\widetilde c_{e,f^\star}^{(T+1)}}{\widetilde c_{e,f^\star}^{(1)}}} +
    \log \paren{\frac{1}{T}\sum_{t \in \sqb{T}}\sum_{g \in S_e}\widetilde \ell_{e,g}^{(t)}} +
    \log \parens{2\parens{1+\alpha_3}} \\
  & \overset{(d)} {\leq} \frac{1}{T} \log r + \log \paren{\frac{1}{T}\sum_{t \in \sqb{T}}
    \sum_{g \in S_e}\widetilde \ell_{e,g}^{(t)}}+\log\parens{2\parens{1+\alpha_3}} \\
  & \overset{(e)}{=} \log \paren{z_e}
  \end{align*}
  where inequality $(a)$ holds since $\log \paren{R_f (c)}$ is convex with
  respect to $c$ (see~\cref{le:effective-resistance-convexity}), equality $ (b)$
  follows from definition of $\mathsf{WeightCompute}$
  (\cref{eq:weight-compute}), inequality $(c)$ follows from the concavity of
  $\log$, and inequality $(d)$ arises from the fact that
  \[
  \frac{\widetilde c_{e,f^\star}^{(T+1)}}{\widetilde c_{e,f^\star}^{(1)}}\leq
  \frac{w_e}{{w_e}/ \parens{\abss{e} -1}}\leq r,
  \]
  the last equality $(e)$ follows directly from the definition of $z_e$ in
  \cref{prop:overestimate-preparation}, with the parameter $C$ selected as in
  line 8 of \cref{alg:QHLSO}.
\end{proof}

\begin{proof}[Proof of \cref{thm:overestimate-quantum}]
  By taking $\alpha_1=\alpha_2=0.1$ as constants and $T =\log r $, $z$ becomes
  a $4n$-overestimate according to \cref{prop:claim-QHLSO}.
  It remains to analyze the time complexity of the algorithm.

  First, we note that the $\mathsf{WeightInitialize}(\cO_H)$ procedure runs in
  $\widetilde O \parens{1}$ time, as established in \cref{prop:weight-init}.
  In each round $t \in \sqb{T}$, $\mathsf{GraphSparsify}(U_{G (t)}, \alpha_1)$
  is executed in $\widetilde O\parens{r\sqrt{mnr}}$ time, following
  \cref{thm:apers-spectral-sparse}, where $\widetilde O\parens {r}$ accounts for
  the query cost of $U_{G (t)}$.
  The resulting graph $\widetilde G^{(t)} $ is a $\alpha_1$-spectral sparsifier
  of $G^{(t)}$, and a crucial fact is that $\widetilde G^{(t )}$ is sparse and
  the number of edges in $\widetilde G^{(t)}$ is $\widetilde O\parens{n}$.
  Hence, we can initialize the data structures
  $\mathsf{UGraphStore}\parens{\widetilde G^{(t )}}$ and
  $ \mathsf{EffectiveResistance}(\widetilde G^{(t)}, \alpha_2) $ in
  $\widetilde O\parens{n}$ time, as per
  \cref{prop:quantum-effective-resistance-oracle} and
  \cref{prop:quantum-underlying-graph-store}.
  The procedure $\mathsf{WeightCompute} \parens{\cO_H, \cR^{(t)}, \cG ^{(t)}}$
  provides $U_{G (t+1)}$, with each query taking $\widetilde O\parens{r}$ time
  according to \cref{prop:weight-compute}.
  In the final step of the algorithm, $\cZ$ can be initialized in
  $\widetilde O \parens{n}$ time and each $\cZ.\mathsf{Query}$ can be executed
  in $\widetilde O\parens{r} $ time, following
  \cref{prop:overestimate-preparation} with $\tau_t =\widetilde O\parens {n}$
  and $\iota_t = \widetilde O\parens {1}$.

  To summarize, the total preprocessing time is
  $\widetilde O\parens{r \sqrt{mnr}}$, and the per-query time is
  $\widetilde O\parens{r}$.
\end{proof}

\section{Proof of \cref{thm:quantum-hypergraph-sparsification}}\label{sec:proof-of-hypergraph-sparsification}

The proof of correctness for the algorithm follows closely the chaining proofs
in~\citet{lee2023spectral} and~\citet{jambulapati2023chaining}. In particular,
we rely on the following crucial technical bound from~\citet{lee2023spectral},
which is derived using Talagrand's generic chaining method.

\begin{lemma}[Corollary 2.13 in~\citet{lee2023spectral}]\label{le:core-lemma-in-lee2023spectral}
  Let $A: \mathbb{R}^n \to \mathbb{R}^s$ be a linear map, with
  $a_1,\ldots ,a_s$ representing the rows of $A$.
  The functions $\phi_1,\ldots,\phi_m:\mathbb{R}^s \to \mathbb{R}$ are in the
  form of $\phi_i \parens{x}=\max_{j \in S_i} w_i\abss{\angs{a_j,x}}$ for some
  $S_i \subseteq \sqb{s}$ and $w \in \sqb{0,1}^{S_i}$.
  Let $D =\max_{i \in \sqb{m}}\abs{S_i}$.
  Then, for any $T\subseteq B^n_2:=\sets{x \in \mathbb{R}^n : \norm{x}^2\leq 1}$, the
  following inequality holds:
  \begin{equation*}
    \E \sup_{x \in T} \sum_{j =1}^m \xi_j \phi_{j}\parens{x}^2 \leq C_{0} \cdot
    \norm{A}_{2 \to \infty} \sqrt{\log\parens{s+n}\log D} \cdot
    \sup_{x \in T} \paren{\sum_{j =1}^m \phi_{j}\parens{x}^2}^{1/2},
  \end{equation*}
  for some universal constant $C_{0}$.
  The variables $\xi_1,\ldots,\xi_m$ are i.i.d.
  Bernoulli random variables taking values $\pm1$, and
  $\norms{A}_{2 \to \infty}$ is defined by
  $\norms{A}_{2 \to \infty}:=\max\sets{\norm{Ax}_{\infty}: x\in B^n_2}$.
\end{lemma}

To provide clarity on how this lemma applies, we explain its connection to
hypergraphs. Let $H = (V, E, w)$ be a hypergraph with $\abss{E}=m, \abss{V}=n$,
and weights $w \in \mathbb{R}^E_{\geq 0}$. The rank of the hypergraph is $r =\max_{e
    \in E}\abss{e}$. In the lemma, $n$ and $m$ correspond to the number of
vertices and hyperedges, respectively. The functions $\phi_i$ capture the
maximization over all edges in the clique generated by replacing each
hyperedge, corresponding to the energy of the hyperedges. The number $s$ refers
to the number of edges in the complete graph $K_n$, i.e., $s = n(n-1)/2$. The
parameter $D$ represents the maximum number of edges in the clique generated by
replacing each hyperedge, i.e., $D = r(r-1)/2$.

Before proving the \cref{thm:quantum-hypergraph-sparsification}, we
present the following fact.
\begin{proposition}\label{prop:energy-comparison}
  Let $H=\parens{V,E,w}$ be a hypergraph with $n$ vertices, and let
  $G=\parens{V,F,c}$ represent its underlying graph.
  For any $x \in \mathbb{R}^n$ such that $x\perp 1$, the inequality
  $Q_H\parens{L_G^{+ /2}x}\geq \norm{x}^2$ holds.
\end{proposition}

\begin{proof}
  For any $v \in \mathbb{R}^n$, we have
  \begin{align*}
    Q_H\paren{v}
    & = \sum_{e \in E} w_e \max_{\sets{i,j} \subseteq e}\parens{v_i -v_j}^2 \\
    & = \sum_{e \in E} \parens{\sum_{\sets{i,j}\subseteq e} c_{ij}^e}
      \max_{\sets{i,j} \subseteq e}\parens{v_i -v_j}^2\\
    & \geq \sum_{e \in E} \sum_{\sets{i,j}\subseteq e}  c_{ij}^e \paren{v_i -v_j}^2 \\
    & = v^\top L_G v.
  \end{align*}
  Taking $v= L_G^{+/2}x$ achieves we desired.
\end{proof}

\begin{proof}[Proof of \cref{thm:quantum-hypergraph-sparsification}]
  Utilizing \cref{alg:QHLSO} we can obtain a query access to $\nu$-overestimate
  $z$ with $\nu=2 \parens{1+\alpha_2} C_1 \cdot n= O\paren{n}$ in
  $\widetilde O\parens{r \sqrt{mnr}}$ time.
  Each query can be executed in $\widetilde O\parens{r}$ time.
  Combining with quantum sampling algorithm (\cref{cor:quantum-prob-sample}), we
  can sample a subset $\widetilde E \subseteq E$ of size
  $M = \widetilde O(n/\epsilon^{2})$ in
  $\widetilde O\parens{\sqrt{Mm} \cdot r} =\widetilde O\parens{r \sqrt{mn}/\epsilon}$
  time.
  The quantum sum estimation (\cref{thm:quantum-sum-estimate}) is executed in
  $\widetilde O\parens{r \sqrt{m}}$ time, obtaining $\widetilde s$ as
  $0.1$-approximation of $s=\norms{z}_1$.
  We then move on to demonstrate the correctness of the algorithm, showing that
  the output $\widetilde H$ is indeed an $\epsilon$-spectral sparsifier of $H$.

  For a hyperegde $e \in E$, we introduce the following notations
  \begin{align*}
    \mu_e & ={z_e}/{\widetilde s}, \\
    a_{ij} & = L_G^{+ /2} \paren{\delta_i -\delta_j}, \quad \forall \set{i,j} \in\sqb{n}, \\
    {a_{ij}^e} & =\sqrt{{\widetilde s w_e}/{z_ {e}}}\cdot  a_{ij}, \quad
                 \forall \set{i,j}\subseteq e, \\
    \phi_ e\paren{x} & =\max \set{\abss{\ang{{a_{ij}^e} , x}}:
                        \forall \set{i,j}\subseteq e}, \quad  \forall x \in\mathbb{R}^n, \\
    \phi^2_ e\paren{x} & =\max \set{\ang{{a_{ij}^e} , x}^2:
                         \forall \set{i,j}\subseteq e},\quad  \forall x \in\mathbb{R}^n,
  \end{align*}
  where $R _e :=\max \sets{R_f: f \in \binom{e}{2}}$.
  Suppose we sample a hyperedge sequence
  $ \sigma_{\mu}= \bigl(e_\mu^{(1)} , \ldots, e_\mu^{(M)} \bigr)$ such that
  each element $e$ is sampled with probability proportional to $\mu_e$ (also
  proportional to $z_e$).
  For the new obtained hypergraph $ H_{\mu}$, the weight of sampled hyperedge
  is
  \begin{equation*}
    w_e^\mu = \frac{\#\set{t\in \sqb{M}:e_\mu^{(t)} =e}}{M}
    \cdot\frac{w_e}{\mu_e},\quad\forall e\in E.
  \end{equation*}
  Recall the definition that
  $Q_e\parens{x}=\max_{\sets{u,v}\subseteq e}\parens{x_u-x_v}^2$, the energy of
  sampled hypergraph $H_\mu$ is given by
  \begin{equation*}
    Q_ {H_\mu}=\frac{1}{M} \sum_{t= 1}^M  \frac{w_{e_\mu^{(t)}}}{\mu_{e_\mu^{(t)}}}
    Q_{e_\mu^{(t)}} \paren{x}.
  \end{equation*}
  We want to choose sample time  $M$ sufficiently large such that
  \begin{equation}
    \underset {H_\mu}{\E}\sqb{\abss{Q_ H \paren{x} -
        Q_{H_\mu}\paren{x}}}\leq \epsilon\cdot Q_H\paren{x},
    \quad\forall x  \in \mathbb{R}^n.
  \end{equation}
  
Equivalently, it suffices to show that
  \begin{equation*}
    \underset {H_\mu}{\E}  \max_{v : Q_H\paren{v}\leq 1}
    \abss{Q_H \paren{v}- Q_{H_\mu}\paren{v}} \leq \epsilon.
  \end{equation*}
  It is worth noting that for $\forall x\in \mathbb{R}^n$,
  \begin{equation*}
    \frac{w_e}{\mu_e}Q_e\paren{L^{+/2}x}=\frac{w_e}{z_e}
    \cdot\widetilde s \max_{\set{i,j}\subseteq e}
    \ang{L^{+ /2}x,\delta_i -\delta_j}^2 =
    \max_{\set {i,j}\subseteq e}\ang{x,a_{ij}^e}^2=\phi_e^2\paren{x}.
  \end{equation*}
  Then we have
  \begin{equation}
    \label{eq:main-0}
    Q_{H_\mu}\paren{L^{+ /2}x}= \frac{1}{M}\sum_{j\in \sqb{M}}\phi_j^2\paren{x}
  \end{equation}
  where $\phi_j\paren{x}$ corresponds to the $j$-th sampled hyperedge using
  $\mu$.
  Let $H_\mu^\prime$ be an independent copy of $ H_\mu $, and
  $\xi_t ,t\in\sqb{M}$ be the i.i.d.
  Bernoulli $\pm 1$ random variables.
  Note that
  $\E_{H_\mu}\sqb{Q_{H_\mu}\paren{x}}=\frac{\widetilde s}{s} Q_{H}\paren{x}$
  and $\widetilde s/s \in \sqb{0.9, 1.1}$.
  By convexity of the absolute value function, for any $x\in T:=\set{x\in \mathbb{R}^n:Q_H \paren{L^{+/2}x}\leq 1}$, we
  have
  \begin{equation*}
    \begin{aligned}
      & \underset {H_\mu}{\E}\max_{x \in T} {\abss{Q_ H \paren{L^{+/2}x}
        - Q_{H_\mu}\paren{L^{+/2}x}}} \\
      & \overset{\textup{concavity}}{\leq} \underset {H_\mu ^\prime}{\E}
        \underset {H_\mu}{\E} \max_{x \in T}{\abss{\frac{s}{\widetilde s}
        Q_ { H_\mu^\prime} \paren{L^{+/2}x} -Q_{H_\mu}\paren{L^{+/2}x}}} \\
      & \leq \frac{s}{\widetilde s} \cdot
        \underset {H_\mu ^\prime}{\E}  \underset {H_\mu}{\E}\max_{x \in T}
        {\abss{Q_ { H_\mu^\prime} \paren{L^{+/2}x}
        -Q_{H_\mu}\paren{L^{+/2}x}}}+
        \paren{\frac{s}{\widetilde s}-1}
        \underset {H_\mu}{\E}\max_{x \in T}{\abss{Q_{H_\mu}\paren{L^{+/2}x}}} \\
      & \overset{\textup{\cref{eq:main-0}}}{=} \frac{s}{\widetilde s}\cdot
        \underset {H_\mu ^\prime}{\E}  \underset {H_\mu}{\E}\max_{x \in T}
        {\abss{\frac{1}{M} \sum_{t\in \sqb{M}} \phi_{e_\mu^{(t)}}^2
        \paren{x} - \phi_{{e_{\mu}^{(t)}}^\prime}^2\paren{x}}}+
        \paren{\frac{s}{\widetilde s}-1}\underset{H_\mu}{\E}\max_{x \in T}
        {\abss{\frac{1}{M} \sum_{t\in \sqb{M}} \phi_{e_\mu^{(t)}}^2\paren{x}}} \\
      & =\frac{s}{\widetilde s} \cdot\underset{\xi}{\E}
        \underset {H_\mu^\prime}{\E} \underset {H_\mu}{\E}\max_{x \in T}
        {\abss{\frac{1}{M}  \sum_{t\in \sqb{M}} \xi_t\paren{\phi_{e_\mu^{(t)}}^2\paren{x} -
        \phi_{{e_{\mu}^{(t)}}^ \prime}^2\paren{x}}}}+
        \paren{\frac{s}{\widetilde s}-1}\underset{\xi}{\E}\underset {H_\mu}{\E}\max_{x \in T}
        {\abss{\frac{1}{M} \sum_{t\in \sqb{M}} \xi_t\phi_{e_\mu^{(t)}}^2\paren{x}}} \\
      & \leq  \frac{2s}{\widetilde s}   \cdot\underset{\xi}{\E}
        \underset {H_\mu}{\E}\max_{x \in T} {\abss{\frac{1}{M} \sum_{t\in \sqb{M}} \xi_t
        \phi^2_ {e_\mu^{(t)}}\paren{x}}}+
        \paren{\frac{s}{\widetilde s}-1}\underset{\xi}{\E}
        \underset {H_\mu}{\E}\max_{x \in T}{\abss{\frac{1}{M}\sum_{t\in \sqb{M}} \xi_t
        \phi^2_ {e_\mu^{(t)}}\paren{x}}} \\
      &= \paren{\frac{3s}{\widetilde s}-1}\cdot \underset{\xi}{\E}
        \underset {H_\mu}{\E} \max_{x \in T} \abss{{\frac{1}{M}   \sum_{t\in \sqb{M}} \xi_t
        \phi^2_ {e_\mu^{(t)}}\paren{x}}}.
    \end{aligned}
  \end{equation*}
Note that for any function $f$, the inequality $\max_{x\in T} |f(x)| \leq \max\{\max _{x\in T}f(x), 0\} + \max\{\max_{x\in T} -f(x), 0\}$ holds. Now, let $f\parens{x}=\frac{1}{M}\sum_{t\in \sqb{M}}\xi_t  \phi^2_ {e_\mu^{(t)}}\paren{x}$. The second term $0$ in $\max\sets{\cdot,0}$, can be attained by the first term, since $\lim_{v\to 1}Q_{H}\parens{v}=0$, combined with the identity \cref{eq:main-0}. Consequently, we have 
  \begin{equation}
    \label{eq:main-1}
    \begin{aligned}
      & \underset {H_\mu}{\E}\max_{x \in T} {\abss{Q_ H \paren{L^{+/2}x}
        - Q_{H_\mu}\paren{L^{+/2}x}}} \\
        &\leq  \paren{\frac{3s}{\widetilde s}-1}\cdot \paren{\underset{\xi}{\E}
        \underset {H_\mu}{\E}\max_{x \in T}  f\parens{x} + \underset{\xi}{\E}
        \underset {H_\mu}{\E}\max_{x \in T}  -f \parens{x}}\\
      &= 2\paren{\frac{3s}{\widetilde s}-1}\cdot \underset{\xi}{\E}
        \underset {H_\mu}{\E} \max_{x \in T} {{\frac{1}{M}   \sum_{t\in \sqb{M}} \xi_t
        \phi^2_ {e_\mu^{(t)}}\paren{x}}}.
    \end{aligned}
  \end{equation}
The last equality holds because $\xi_t, t \in \sqb{T}$ are i.i.d. Bernoulli $\pm 1$ random variables.
  
Consider the random process
  $V_x=\frac{1}{M}\sum_{j \in \sqb{M}}\xi_j \phi^2_j\paren{x}$, where
  $\phi_j\paren{x}$ corresponds to the $j$-th sampled hyperedge using $\mu$,
  $x\in T$. By applying \cref{le:core-lemma-in-lee2023spectral} with the linear map $A:\mathbb{R}^n\to \mathbb{R}^{n(n-1)/2}$ defined as 
  \begin{equation*}
  \parens{Ax}_{ij}:=\max_{e\in E:\sets{i,j}\in \binom{e}{2}}\norms{a_{ij}^e}\cdot\ang{x,{a_{ij}}/{\norms{a_{ij}}}},
  \end{equation*} and using  \cref{prop:energy-comparison} to ensure that $T \subseteq B_2^n$, as required by \cref{le:core-lemma-in-lee2023spectral},
  we obtain:
  \begin{equation}
    \label{eq:main-2}
    \begin{aligned}
      \underset{\xi}{\E}\sup_{x \in T} V_x
      & \leq C_{0} \cdot \frac{\norms{A}_{2\to \infty}\cdot
        \sqrt{\log \paren{n (n -1)/2 +n} \log \paren{r (r-1)/2}}}{\sqrt {M}}
        \max_{x \in T}\paren{\frac{1}{M}\sum_{j \in \sqb{M}}   \phi _j^2 \paren{x}}^{1/2} \\
      & \leq 2 C_{0}\cdot \frac{\norms{A}_{2\to \infty}\cdot
        \sqrt{\log n \log r}}{\sqrt {M}}\max_{x \in T}
        \paren{\frac{1}{M}\sum_{j \in \sqb{M}} \phi _j^2 \paren{x}}^{1/2}
    \end{aligned}
  \end{equation}
  where
  \begin{align*}
    \norms{A}_{2\to \infty}
    & =\max\sets{\norm{Ax}_{\infty}: x\in B^n_2} =\max_{e \in E} \max_{\sets{i,j}\subseteq e}\norm{a_{ij}^e} \\
    & = \max_{e \in E} \max_{\sets{i,j}\subseteq e}
      \sqrt{\frac{\widetilde s w_{e}}{z_e}\cdot R_{ij}} = \sqrt{\widetilde s}\cdot \max_{e \in E} \max_{\sets{i,j}\subseteq e}
      \sqrt{\frac{w_{e} R_{ij}}{z_e} }\\
    &\leq \sqrt{\widetilde s}.
  \end{align*}
  Observe that the component in RHS  of \cref{eq:main-2} can be written as
  \begin{equation*}
    \max_{x \in T} \frac{1}{M}\sum_{j \in \sqb{M}}   \phi _j^2 \paren{x} =
    \max_{v : Q_H \paren{v}\leq 1} \frac{1}{M}\sum_{j \in \sqb{M}} \phi _j^2
    \paren{L^{+/2} v}=\max_{v :Q_{H} \paren{v}\leq 1} Q_{H_\mu}\paren{v}.
  \end{equation*}
The first equality holds because $Q_H\parens{x}=Q_H\parens{x^\prime}$ whenever $x-x^\prime \in \operatorname{ker}\parens{L_G}$.
  Then we have
  \begin{equation*}
    \begin{aligned}
      \tau :=\underset {H_\mu}{\E}  \max_{v : Q_H\paren{v}\leq 1}
      \abss{Q_H \paren{v}- Q_{H_\mu}\paren{v}}
      & = \underset {H_\mu}{\E}  \max_{x \in T}
        \abss{Q_H \paren{L^{+ /2}x}- Q_{H_\mu}\paren{L^{+ /2} x}} \\
      & \overset{\textup{\cref{eq:main-1}}} {\leq} \paren{\frac{6s}{\widetilde s}-2}
        \cdot\underset{\epsilon}{\E}  \underset {H_\mu}{\E}
        \max_{x \in T}{{\frac{1}{M}   \sum_{t\in \sqb{M}} \epsilon_t
        \phi^2_ {e_\mu^{(t)}}\paren{x}}} \\
      & \overset{\textup{\cref{eq:main-2}}} {\leq}4\paren{\frac{3s}{\widetilde s}-1}
        \cdot C_{0} \cdot {\sqrt{\widetilde s   \log n \log r/M}}\cdot
        \underset {H_\mu} {\E}\paren{\max_{v :Q_{H} \paren{v}\leq 1} Q_{H_\mu}\paren{v}}^{1/2} \\
      & \overset{\textup{concavity}}{\leq}4\paren{\frac{3s}{\widetilde s}-1}
        \cdot C_{0} \cdot { \sqrt{\widetilde s     \log n \log r /M}}\cdot
        \paren{\underset {H_\mu} {\E} \max_{v :Q_{H} \paren{v}\leq 1} Q_{H_\mu}\paren{v}}^{1/2} \\
      & \leq 4\paren{\frac{3s}{\widetilde s}-1}\cdot C_{0}\cdot
        \sqrt{\widetilde s    \log n \log r /M}\cdot \paren{1+\tau}^{1/2} \\
      & \leq 4\paren{\frac{3s}{\widetilde s}-1}  \cdot C_{0} \cdot
        { \sqrt{\widetilde s    \log n \log r /M}} \cdot \paren{1+ \frac{1}{2}\tau} \\
      & \leq  10  C_{0} \cdot { \sqrt{\widetilde s    \log n \log r /M}}
        \cdot \paren{1+ \frac{1}{2}\tau}.
    \end{aligned}
  \end{equation*}
  Therefore, we have $\tau \leq 20 C_{0} \sqrt{\widetilde s \log n \log r /M}$
  whenever $M \geq 100 C_{0}^{2}\widetilde s \log n \log r$.
  Choosing
  $M := 400 C^2_{0} {\widetilde s {\log n \log r}}/{\epsilon^2}=
  \Theta\parens{n \log n \log r /\epsilon^2}$
  yields
  \begin{equation*}
    \underset {H_\mu}{\E}  \max_{v : Q_H\paren{v}\leq 1}
    \abss{Q_H \paren{v}- Q_{H_\mu}\paren{v}}=\tau \leq \epsilon.
  \end{equation*}

\end{proof}


\end{document}